\documentclass[11pt]{article}

\usepackage{fullpage}
\usepackage{array,fullpage,multirow}
\usepackage{amssymb,amsmath,amsthm,sectsty,url,nicefrac,color}
\usepackage{hyperref}
\usepackage{cleveref}
\usepackage{aliascnt}
\usepackage[numbers]{natbib} 

\usepackage{algpseudocode}

\numberwithin{equation}{section}

\newcommand{\remove}[1]{}

\usepackage[american]{babel}
\usepackage{amsmath}
\usepackage{amsthm}
\usepackage{amssymb}
\usepackage{graphicx}
\usepackage{xargs}
\usepackage{hyperref}
\usepackage{multirow}
\usepackage{array}
\usepackage{extarrows}
\usepackage{ifthen}
\usepackage{mathtools}
\usepackage{dsfont}
\usepackage{verbatim}
\usepackage{pgf}
\usepackage{tikz}
\usepackage{pgfplots}
\usepackage[ruled,boxed,linesnumbered]{algorithm2e}
\usepackage[justification=centering]{caption}
\usepackage{enumitem}
\usepackage{float}
\usepackage{pdfpages}

\pgfplotsset{compat=1.14}
\usepgfplotslibrary{fillbetween}
\usetikzlibrary{fit}
\usetikzlibrary{arrows,automata}
\usepackage{tkz-graph}
\usetikzlibrary{shapes.geometric}%

\newtheorem{theorem}{Theorem}[section]
\newtheorem{lemma}[theorem]{Lemma}

\newtheorem{fact}[theorem]{Fact}
\newtheorem{corollary}[theorem]{Corollary}
\newtheorem{definition}[theorem]{Definition}

\newcommand{\pr}[1]{Pr\left[#1 \right] }

\newcommand{\pimul}[4][]{ \ifthenelse{\equal{#2}{}}{\Pi#1}{\ifthenelse{\equal{#3}{}}{\underset{#2}{\prod}#1}{\ifthenelse{\equal{#4}{}}{\underset{#2\in #3}{\prod}#1}{\underset{#3\le #2\le #4}{\prod}#1}}}}
\newcommand{\sigsum}[4][]{ \ifthenelse{\equal{#2}{}}{\Sigma#1}{\ifthenelse{\equal{#3}{}}{\underset{#2}{\sum}#1}{\ifthenelse{\equal{#4}{}}{\underset{#2\in #3}{\sum}#1}{\underset{#3\le #2\le #4}{\sum}#1}}}}
\newcommand{\mech}[1][M]{{#1} }

\newcommand{\chainmech}[3][M]{\mathcal{#1}_{#3\circ#2}}
\newcommand{\multichainmech}[3][M]{\mathcal{#1}_{\underset{#3}{\bigcirc}#2}}

\DeclareMathOperator{\Ber}{Ber}

\newcommand{\exv}[1][]{\ifthenelse{\equal{#1}{}}{\mathbb{E}}{\mathbb{E}\left[#1\right]}}
\newcommand{\ber}[1][p]{\Ber\left(#1\right)}

\DeclarePairedDelimiter\Norm{\big\lVert}{\big\rVert}
\DeclarePairedDelimiter\abs{\big\lvert}{\big\rvert}
\DeclarePairedDelimiter\floor{\lfloor}{\rfloor}
\DeclarePairedDelimiter\ceil{\lceil}{\rceil}

\newcommand{\norm}[2][]{\ifthenelse{\equal{#1}{}}{\Norm{#2}}{\Norm{#2}_{#1}}}

\newcommand{\ldp}[1][$\varepsilon$]{#1-LDP}

\newcommand{\todo}[1]{}

\title{Can Two Walk Together: Privacy Enhancing Methods and Preventing Tracking of Users}
\author{
	Moni Naor\thanks{Department of Computer Science and Applied Mathematics,
		Weizmann Institute of Science,  Rehovot 76100, Israel. Email:
		\texttt{moni.naor@weizmann.ac.il}. Supported in part by grant  from the Israel
		Science Foundation (no.\ 950/16). Incumbent of the Judith Kleeman Professorial
		Chair.}
	\and
	Neil Vexler\thanks{Department of Computer Science and Applied Mathematics,
		Weizmann Institute of Science,  Rehovot 76100, Israel.}
}

\begin{document}
	\maketitle
	\begin{abstract}
		We present a new concern when collecting data from individuals that arises from the attempt to mitigate privacy
		leakage in multiple reporting: tracking of users participating in the data collection via the mechanisms added to provide privacy. We present several definitions for untrackable mechanisms, inspired by the differential privacy framework.
		
		Specifically, we define the trackable parameter as the log of the maximum ratio between the probability that a set of reports originated from a single user and the probability that the same set of reports originated from two users (with the same private value).
		We explore the implications of this new definition.
		We show how differentially private and untrackable mechanisms  can be combined to achieve a bound for the problem of detecting when a certain user changed their private value.
		
		Examining Google's deployed solution for everlasting privacy, we show that RAPPOR (Erlingsson et al.\ ACM CCS, 2014) is trackable in our framework for the parameters presented in their paper.
		
		We analyze a variant of randomized response for collecting statistics of single bits, Bitwise Everlasting Privacy, that achieves good accuracy and everlasting privacy, while only being reasonably  untrackable, specifically grows linearly in the number of reports. For  collecting statistics about data from larger domains (for histograms and heavy hitters) we present a mechanism that prevents tracking for a limited number of responses.
		
		We also present the concept of Mechanism Chaining, using the output of one mechanism as the input of another, in the scope of Differential Privacy, and show that the chaining of an $\varepsilon_1$-LDP mechanism with an $\varepsilon_2$-LDP mechanism is $\ln\frac{e^{\varepsilon_1+\varepsilon_2}+1}{e^{\varepsilon_1}+e^{\varepsilon_2}}$-LDP and that this bound is tight.
	\end{abstract}
	\section{Introduction}
	The cure should not be worse than the disease. In this paper we raise the issue that mechanisms for Differentially Private data collection  enable the tracking of users. This wouldn't be the first time an innocent solution for an important problem is exploited for the purposes of tracking. Web cookies, designed to let users maintain a session between different web pages, is now the basis of many user tracking implementations. In the Differential Privacy world, we examine how various solutions meant to protect the privacy of users over long periods of time actually enable the tracking of participants.
	
To better understand this, consider the following scenario: A browser developer might wish to learn what which are the most common homepages, for caching purposes, or perhaps to identify suspiciously popular homepages that might be an evidence for the spreading of a new virus. They develop a mechanism for collecting the URLs of users' homepages. Being very privacy aware, they also make sure that the data sent back to them is Differentially Private. They want to ensure they can collect this data twice a day without allowing someone with access to the reports to figure out the homepage of any individual user.
	
If fresh randomness is used to generate each differentially private report, then  the danger is that information about the users homepage would be revealed eventually to someone who follows the user's reports. We strive to what we call ``Everlasting Privacy",  the property  of maintaining privacy no matter how many collections were made. In our example, the users achieve everlasting privacy by correlating the answers given at each collection time: e.g.\ a simple way is that each user fixes the randomness they use, and so sends the same report at each collection.
	
	Now consider Alice, a user who reports from her work place during the day and from her home during the evening. At every collection, Alice always reports regarding the same homepage\footnote{The reader may be wondering why bother reporting about the same value if it does not change. For instance it may for purposes of aggregating information about the currently online population.},  and therefore (since the randomness was fixed) sends identical reports at home and at work. An eavesdropper examining a report from the work IP address and a report from Alice's home IP address would notice that they are the same, while if they examined a report generated by Alice and one generated by Bob (with the same homepage) they will very likely be different. This allows the adversary to find out where Alice lives. 
	
	To elaborate, correlation based solutions open the door to the new kind of issue, tracking users.  The correlation between reports can be used as an instrument of identifying individuals, in particular it makes the decision problem of whether or not two sets of reports originated from the same user much easier. This concern has been suggested by the RAPPOR project~\cite{Erlingsson2014a} but without a formal definition, or analysis in the framework where their solution was provided.
	
	The problem of tracking users is related to the problem of point change detection, i.e.\ identifying when a stream of samples switched from one distribution to another. While this problem has been researched in the past under the lens of privacy by Cummings et al.~\cite{Cummings2018,Cummings2019}, these works focused on private release of point change detection, i.e.\ how to enable researchers to detect changes in the sampled distribution while not being too reliant on any specific sample. Our goal is different. We wish to {\em prevent} change point detection as much as we can; as in our case, a change in distribution correlates to a change in private value. Detecting a change in private value jeopardizes the privacy of the user (think of a case where the gender is changed).
	
\subsection{Our Contributions}
	The main conceptual contribution of this work is the definition of reports being untrackable, presented in Section \ref{Stateful mechanisms and Tracking section}. Roughly, the definition states that the distribution on outputs generated by a single user needs to be sufficiently close to that generated by two users. For the discussion on motivation and possible variants see Section~\ref{untrackable discussion subsection}
	\begin{definition}[informal]
		A mechanism $\mech$ is $\left(\gamma,\delta\right)$-Untrackable for $k$ reports, if for any $k$ reports
		\begin{equation*}
		\pr{\text{Reports were generated by one user}}\le e^{\gamma}\pr{\text{Reports were generated by two users}}+\delta
		\end{equation*}
	\end{definition}
	We present a formal definition to Everlasting Privacy. Roughly speaking, a mechanism is $\left(\gamma,\delta\right)$-Everlasting Privacy if executing it any number of times is $\left(\gamma,\delta\right)$-DP. Our main goal is to simultaneously achieve both tracking prevention and everlasting privacy, while maintaining a reasonable accuracy for the global statistics. We explore the implications of this new definition, specifically how it composes and what a fixed state that is reported in a noisy manner can achieve.
	
We describe how our tracking definitions can be extended to the change point detection framework, namely to bound the probability that a change in the user's private value is ever detected. In that section we also discuss the necessity of correlating answers between data collections to ensure Differential Privacy, and define various general constructions for mechanisms that can achieve this Everlasting Differential Privacy.
	
As a tool for analyzing such constructions, in Section~\ref{mechanism chaining section} we prove a theorem about running a Local Differential Privacy mechanism on the output of another such mechanism.
	\begin{theorem}[informal]
		A mechanism that consists of running an $\varepsilon_2$-LDP mechanism on the result of an $\varepsilon_1$-LDP mechanism results in $\frac{1}{2}\varepsilon_1\cdot\varepsilon_2$-LDP for small $\varepsilon_1$ and $\varepsilon_2$.
	\end{theorem}
	Theorem~\ref{advanced chaining} and Corollary~\ref{advanced chaining simpler corollary} provide the formal statement and proof.
	
	We then continue to analyze Google RAPPOR's \cite{Erlingsson2014a} performance under the framework of tracking. We show the pure tracking bound RAPPOR achieves as well as estimate its ``average" case performance. We cocnslude that according to our definition of untrackable, RAPPOR achieves poor protection guarantees. This is presented in Section~\ref{Untrackable in RAPPOR section}.
	
	As a warm up, in Section~\ref{Bitwise Everlasting Privacy Mechanism section} we present a mechanism that deals with data collection of a single private bit from each participant. One can view it as the extension of randomized response in this setting.  Each user generates a bit at random and remembers it. At each collection, the user generates an new bit and sends the XOR of the private bit, the remembered bit and the new bit. The remembered bit is generated by flipping one biased coin, parameterized by $\varepsilon_1$. The new bits are generated from fresh coin flips from another biased coin, parameterized by $\varepsilon_2$. The aggregator collects all the reports and outputs estimated frequencies for both $0$ and $1$. We prove that for a choice of privacy parameters $\varepsilon_1,\varepsilon_2<1$, and for $n$ participating users, the mechanism has the properties:
	\begin{enumerate}[label=(\roman*)]
		\item It is $\varepsilon_1$-Everlasting Differentially Private.
		\item Accuracy: the frequency estimation of $0$ and $1$ is no further than $\tilde{O}\left(\frac{1}{\varepsilon_1\cdot\varepsilon_2\cdot\sqrt{n}}\right)$ from the actual values.
		\item It is $\floor{\frac{k}{2}}\varepsilon_2$-untrackable for $k$ reports.
	\end{enumerate}
	
In Section~\ref{Report noisy inner product section} we present a mechanism that allows the collection of statistics of users private values when their data is $d$ bits. This mechanism is particularly relevant for the problems of heavy hitters estimations and histograms. The mechanism's state consists of the results of the inner product of the private value with multiple vectors in a way that is Differentially Private, reporting one such vector and the private result of the inner product at each data collection. The aggregator collects all the reports and produces an estimate for the frequencies of all possible values, such that the sum of frequencies is $1$. We prove that for a choice of privacy parameters $\varepsilon<1$, setting the state to consist of $L$ reports, and for $n$ participating users, the mechanism has the properties:
	\begin{enumerate}[label=(\roman*)]
		\item It is $\left(\varepsilon,\delta\right)$-Approximate Everlasting Differentially Private.
		\item The estimation of the frequency of all values is no further than $\tilde{O}\left(\frac{1}{\varepsilon^\prime}\sqrt{\frac{d}{n}}\right)$ from the actual frequency, for $\varepsilon^\prime=\frac{\varepsilon}{2\sqrt{2L\ln\left(\frac{1}{\delta}\right)}}$.
		\item It is $\left(0,\frac{k^2}{L}\right)$-untrackable.
	\end{enumerate}
Concretely, to obtain  $\left(\varepsilon,\delta\right)$-Everlasting Privacy and $\alpha$ accuracy, then for $k$ reports the guarantee on the mechanism is $\left(0, \widetilde{O}\left(\frac{k^2}{\alpha^2\varepsilon^2n}\right)\right)\text{-Untrackable}$.

Coming up with better bounds or showing the inherent  limitations is the main open direction we propose (see Section~\ref{sec:open}). 
\section{Preliminaries}\label{preliminaries section}
	
\subsection{Differential Privacy}
For background on Differential Privacy see Dwork and Roth \cite{DworkR14} or Vadhan~\cite{Vadhan2017}.
	
Throughout most of this paper we consider a variant of Differential Privacy, called {\em Local Differential Privacy}. Local Differential Privacy regards mechanisms where each individual user runs on their own data to create a report, which is then sent to the server and aggregated there to produce a population level result. The setting we consider is one where the aggregator access the users' data only through a randomized mapping, a mechanism, that has the following property:
\begin{definition}[\cite{Kasiviswanathan2008}]\label{local differential privacy}
Let $\varepsilon,\delta>0$. A mechanism $\mech:U\mapsto O$ is $\left(\varepsilon,\delta\right)$-Local Differentially Private if for every two possible inputs, $u,u^\prime\in U$, and $\forall S\subseteq O$, $\pr{\mech\left(u\right)\in S}\le e^\varepsilon\cdot\pr{\mech\left(u^\prime\right)\in S}+\delta$
\end{definition}
	
One of  the significant properties of Differential Privacy is the way it composes. Composing two mechanisms that are $\left(\varepsilon_1, \delta_1\right)$ and $\left(\varepsilon_2, \delta_2\right)$-Differentially Private respectively is $\left(\varepsilon_1+\varepsilon_2,\delta_1+\delta_2\right)$-Differentially Private. A small deterioration in the $\delta$ parameter achieves a great improvement in the $\varepsilon$ parameter of the composition.
\begin{theorem}[Advanced composition for Differential Privacy \cite{DworkRV2010}]\label{advanced composition for differential privacy}
		Let $\delta^\prime>0$. The $k$ fold composition of $\left(\varepsilon,\delta\right)$-Differentially Private mechanisms is $\left(\varepsilon^\prime,k\delta+\delta^\prime\right)$-Differentially Private for $\varepsilon^\prime=\sqrt{2k\ln\left(1/\delta^\prime\right)}\varepsilon+k\varepsilon\left(e^\varepsilon-1\right)$
\end{theorem}

Another useful property of Differential Privacy is that running any function on the output of an $\left(\varepsilon, \delta\right)$-Differentially Private mechanism is $\left(\varepsilon, \delta\right)$-Differentially Private. That is, Differential Privacy is {\em closed under post-processing}.

When using the same mechanism to collect reports multiple times, if not done carefully, the privacy guarantee might deteriorate as the number of collections periods grows. We define {\em Everlasting Differential Privacy} as an upper bound on the privacy parameter of a mechanism, no matter how many times it is executed, as long as the private data had not changed. Definition~\ref{everlasting privacy definition} formalizes this idea.

	\subsection{Background}\label{background subsection}
The need for everlasting privacy became apparent since the early stages of the Differential Privacy research. As mentioned in Section \ref{preliminaries section}, independent repetitive executions of Differential Privacy mechanisms inevitably deteriorate the privacy guarantee. While Theorem \ref{advanced composition for differential privacy} teaches us that the privacy guarantee can grow as low as only the square root of the number of reports, practical implementations might require users to participate in as many as thousands of data collections (e.g. anything requiring daily reports).

This led researchers to suggest data collection mechanisms that allow numerous data collections, while maintaining individuals' privacy. Certain solutions, such as Google's RAPPOR \cite{Erlingsson2014a} and Microsoft's dBitFLip \cite{Ding2017}, use the concept of statefulness, maintaining some data between executions. This enables them to correlate outputs between executions, which allows for a manageable upper bound of the privacy leakage that does not rely on the number of collections made. This effectively allows for a privacy guarantee that holds forever, namely Everlasting Privacy.

\paragraph*{Heavy Hitter Mechanisms:}	
Two problems that have been very interesting for data collectors are the histogram and heavy hitters problems. In the histogram problem the goal is to accurately estimate the frequencies of all possible values the population might hold. The heavy hitters problem is about identifying the most common values amongst the population. Both histograms and heavy hitters in the local model has been researched before by Bassily et al.~\cite{BassilyS2015,BassilyNST2017}, who used Hadamard transformations on the users private data that allow users to send succinct reports to the curator while allowing the required statistics to be generated very efficiently. These works do not fit our framework, as they intrinsically allow for trackability. In their solution, each user is associated with a specific piece of some shared randomness. The aggregator must know to which piece of randomness a specific report belongs to, essentially forcing their solution to be highly trackable. The techniques used in their paper are similar to the ones used by Naor et al.~\cite{NaorPR2018}. In that work the authors use an inner product mechanism to identify and ban the most common passwords. This enables the increase in the effective time an adversary will need to invest in order to guess a user's password. Their mechanism maintains Differential Privacy to prevent the leakage of each individual's password, but it does not maintain Everlasting Privacy. They also mention a modification to their scheme  achieves Everlasting Privacy, by reusing the same random vector for all future inner products, but such a solution is highly trackable. The inner product mechanisms used in ~\cite{BassilyS2015,BassilyNST2017,NaorPR2018} were the inspiration of our Noisy Inner Product mechanism presented in \ref{Report noisy inner product section}.

\paragraph*{Continual Observation and Pan Privacy:}
Other models and solutions to long-lasting privacy have been developed as well, such as the Continual Observation model in~\cite{DworkNPR10,ChanSS11}. The goal is to maintains differential privacy for values that change over time, e.g.\ a counter that updates over time, or streams of data, like traffic conditions and so on. This solution is in the central, or streaming, model and not in the local model. Another model is that of Pan Privacy, where the goal is to maintain privacy even if the internal representation of the secret state is leaked from time to time (Dwork et al.~\cite{DworkNPRY10}).  In Erlignsson et al.~\cite{Erlingsson2019} this idea was extended, transforming the mechanism in~\cite{DworkNPR10} to the local model, in order to solve the 1-bit histogram problem, and thus achieving privacy over extended periods of time. The transformation means that every user reports genuinely only once throughout all data collections, thus resulting in accuracy that relies linearly in the number of times their value changed. This suggests that accuracy will drop as collection times increase.

Joseph et al.~\cite{Joseph2018} suggested an approach where at the beginning, a global update occurs, where each individual participates in a private histogram estimation. At each subsequent potential collection time, each user compares their current contribution to the histogram compared to the last time a global update occurs. Depending on how different it is, they are more likely to suggest that another global update occurs. If enough users vote in favor, the curator initiates another round of global update, creating a more accurate histogram. This solution allows for collections to be made from users only when it is likely that the previously computed output is no longer accurate, greatly increasing the privacy guarantee of individuals. On the other hand, their accuracy analysis relies on the existence of a small number of user types, where all users of the same type behave identically.
	
\section{Stateful Mechanisms and Tracking}\label{Stateful mechanisms and Tracking section}
Consider a mechanism for users to report their values to a center.
Such mechanisms may be	{\em stateless}, i.e.\ ones that receive an input and (probabilistically) produce an output, or {\em stateful} mechanisms, ones that receive in addition to the input  a {\em state} and produce in addition to an output a state for the next execution. The power of stateful mechanisms is that they enable the correlation of outputs between different executions through the states passed from one execution to the next.
	
\subsection{Definitions of Mechanisms and Report Stream Generators}
	Stateless mechanisms are randomized mappings for which each execution is independent of the others. Stateless mechanisms receive the user's data and publicly available information, namely auxiliary information, and output a report. The publicly available information can be anything known to all parties, like time of day, value of some publicly accessible counter, etc.
	\begin{definition}[Stateless Mechanism] \label{stateless mechanism}
		A stateless mechanism $\mech$ is a randomized mapping from a user's data and auxiliary information to the domain of reports, $\mech:U\times A\times\left\{0,1\right\}^\star\mapsto R$. In our setting it is used to generate a stream of reports, $r_1,r_2, \ldots$, where each report is generated independently.
	\end{definition}
	Stateless mechanisms might provide very poor everlasting privacy, as each iteration reveals more information about the user's data.
	
	Therefore, to achieve everlasting privacy one must \emph{correlate} the reports sent by the user(s) (see for instance~\cite{DworkNV2012} where this is proved for counting queries). For this we define \emph{Stateful Mechanisms}, where the mechanism maintains a state that is updated with each call to the mechanism.
	\begin{definition}[Fully Stateful Mechanism] \label{fully stateful mechanism}
		A fully stateful mechanism $\mech$ is a randomized mapping from a user's data, current state and auxiliary information to the domain of reports and to a new state, $\mech:U\times S\times A\times\left\{0,1\right\}^\star\mapsto R\times S$. In our setting it is used to generate a stream of reports $r_1,r_2,\ldots$ and a stream of states $s_0=\bot,s_1,\ldots$, such that each pair of state and report are generated by the previous state, auxiliary information and the user's data, $r_i,s_i=\mech\left(u,s_{i-1},a_{i-1}\right)$
	\end{definition}
	Notice that the execution number and all previous outputs can be encoded into the state. A fully stateful mechanism can achieve everlasting privacy by correlating answers using the data stored in the state. For example, it can execute a DP mechanism on the user's data and remember the result, reporting the same result whenever queried.

One shortcoming of correlating the reports in such a manner is that it   might be used as an identifier by an adversary, potentially allowing the adversary to identify that a group of reports all originated from the same user, thus allowing tracking other activities of the user (see Section~\ref{sec:everlasting_tracking}).

	We define \emph{Permanent State Mechanisms} as mechanisms that maintain the same state once set, i.e.\ $s_1=s_2=s_3...$. As we shall see, such mechanisms are very convenient to work with and have good properties wrt composition.
	
	Report stream generators (RSG) are mappings that use mechanisms to generate a stream of reports. The responsibility of the RSG is to get the user's data and iteratively call the mechanism.
	\begin{definition}[Stateless Report Stream Generator]
For a domain of user data $U$, a range of reports $R$, and a report stream size $n$, A Stateless Report Stream Generator using a stateless mechanism $\mech$ is a mapping $\mech[G]^{\mech}_n:U\mapsto R^n$, that acquires the auxiliary information required at each step and calls $\mech$ to generate the reports $r_1,\ldots, r_n$.
	\end{definition}
Similarly, stateful RGSs use fully stateful mechanisms to generate the stream of reports.
	\begin{definition}[Stateful Report Stream Generator]
		For a domain of user data $U$, a range of reports $R$, and a report stream size $n$, A Stateful Report Stream Generator using a fully stateful mechanism $\mech$ is a mapping $\mech[G]^{\mech}_n:U\mapsto R^n$, that acquires the auxiliary information required at each step and calls $\mech$ with the state of the current step to generate report $r_i$ and the next step's state $s_i$.
	\end{definition}

\subsection{Everlasting Privacy and Tracking}
\label{sec:everlasting_tracking}
The problem we focus on is the ability of an adversary to distinguish whether or not a set of reports originated from a single user or by two users (or more, see Section~\ref{untrackable discussion subsection}). For example, If an adversary had two sets of reports belonging to two different IP addresses, the adversary could learn if those IP addresses belong to the same user or not (potentially identifying the user's work place or home address). The definition of untrackable we propose is inspired by definition of Differential Privacy.
	
\begin{definition} \label{approximately untrackable}
For a domain of user data $U$, a range of reports $R$ and a report stream size $k$, a report stream generator $\mech[G]^{\mech}_k$ is $(\gamma,\delta)$-{\bf untrackable} if for all user data $u\in U$, for all subsets of indices $J\subseteq[k],J^\complement=[k]\setminus J$ and $\forall T\subseteq R^{k}$ we have:
		
		$$\pr{\mech[G]^{\mech}_k\left(u\right)\in T}\le e^\gamma\cdot\pr{\mech[G]^{\mech}_{\left|J\right|}\left(u\right)\in T_J}\cdot\pr{\mech[G]^{\mech}_{k-\left|J\right|}\left(u\right)\in T_{J^\complement}}+\delta$$
and
		$$\pr{\mech[G]^{\mech}_{\left|J\right|}\left(u\right)\in T_J}\cdot\pr{\mech[G]^{\mech}_{k-\left|J\right|}\left(u\right)\in T_{J^\complement}}\le e^\gamma\cdot\pr{\mech[G]^{\mech}_k\left(u\right)\in T}+\delta$$
	\end{definition}
	
	For report stream generators that are $\left(\gamma,\delta\right)$-untrackable, an adversary has  only a small advantage in distinguishing between the following two cases: the reports originated from a single user or two users. A discussion for the idea behind this definition and its benefits can be found in Section~\ref{untrackable discussion subsection}. If we  want this property to hold for any possible output (i.e.\ always have the ambiguity), then we can demand that the mechanism be $\left(\gamma,0\right)$-untrackable. We call such mechanisms $\gamma$-untrackable. We leverage the similarity to DP show  composition theorems on untrackable mechanisms.
	
	Everlasting Privacy is meant to limit the leakage of information users suffer, no matter how many executions a mechanism had. For the following definitions let $T$ be a collection of report streams. For a set of indices $J$ let $T_J$ be the collection of partial report stream, where the reports taken are those in indices $J$.
	\begin{definition}[Everlasting Privacy]\label{everlasting privacy definition}
		
		For a domain of user data $U$, a range of reports $R$, a report stream generator $\mech[G]^{\mech}_k$ is $\left(\varepsilon,\delta\right)$-{\bf Everlasting Privacy} if for all user data $u, u^\prime \in U$,  for all report stream size $k$ and for all sets of output streams $T\subseteq R^{k}$, $\pr{\mech[G]^{\mech}_k\left(u\right)\in T}\le e^{\varepsilon}\pr{\mech[G]^{\mech}_k\left(u^\prime\right)\in T}+\delta$
	\end{definition}
	If a mechanism is $\left(\varepsilon,0\right)$-Everlasting Privacy we say it is $\varepsilon$-Everlasting Privacy.
	
	These definitions are tightly related to the problem of change-point detection. We define undetectability similarly to untrackability, only we do not assume both report sets originated from the same private data:
	\begin{definition} \label{approximately undetectable}
		For a domain of user data $U$, a range of reports $R$ and a report stream size $k$, a report stream generator $\mech[G]^{\mech}_k$ is $\left(\gamma,\delta\right)$-{\bf undetectable} if for all pairs of user data $u,u^\prime\in U$, for all subsets of indices $J\subseteq[k],J^\complement=[k]\setminus J$ and $\forall T\subseteq R^{k}$ we have:
		
		$$\pr{\mech[G]^{\mech}_k\left(u\right)\in T}\le e^\gamma\cdot\pr{\mech[G]^{\mech}_{\left|J\right|}\left(u\right)\in T_J}\cdot\pr{\mech[G]^{\mech}_{k-\left|J\right|}\left(u^\prime\right)\in T_{J^\complement}}+\delta$$ and
		
		$$\pr{\mech[G]^{\mech}_{\left|J\right|}\left(u\right)\in T_J}\cdot\pr{\mech[G]^{\mech}_{k-\left|J\right|}\left(u^\prime\right)\in T_{J^\complement}}\le e^\gamma\cdot\pr{\mech[G]^{\mech}_k\left(u\right)\in T}+\delta.$$
	\end{definition}
	We can now connect being untrackable and everlasting privacy with being undetectable.
	\begin{theorem}\label{undetectable from untrackable and everlasting}
		A mechanism that is $\left(\gamma, \delta\right)$-untrackable and $\left(\varepsilon,\delta^\prime\right)$-everlasting differentially private is also $\left(\gamma+\varepsilon,\delta_{\max}\right)$-undetectable, for $\delta_{\max}=\max\left\{e^{\varepsilon}\delta+\delta^\prime,\delta+e^{\gamma}\delta^\prime\right\}$.
	\end{theorem}
	The proof for this theorem can be found in Appendix~\ref{proof of undetectable from untrackable and everlasting}

	\subsection{Tracking Bounds, Composition Theorems and Generalizations}
	For the special case of Permanent State Mechanisms, we can show an upper bound on the untrackable parameter. If the mechanism is $\varepsilon$-Differentially Private in its state, i.e.\ the mechanism protects the privacy of the state,  then the untrackable parameter grows linearly in $\varepsilon$:
	
	\begin{theorem}\label{untrackable bound for permanent state mechanisms}
		A Permanent state mechanism whose reports are generated by an $\varepsilon$-Differentially Private mechanism receiving the state as its input is $\floor*{\frac{k}{2}}\varepsilon$-untrackable for $k$ reports.
	\end{theorem}
	The proof for this theorem can be found in \ref{proof of permanent state tracking upper bound}
	
	An important question is how tracking composes, i.e.\ how does a user's participation in multiple Report Stream Generators affect his untrackable guarantees. The similarity between the definition of untrackability and differential privacy allows us to apply results regarding the latter to obtain results on the former.  We show an advanced composition for untrackable mechanisms that is analogous to advanced composition for differential privacy~\cite{DworkRV2010} and Theorem~\ref{advanced composition for differential privacy}.
	\begin{theorem}[Advanced composition for untrackability]\label{untrackable advanced composition}
		Let $m$ be a positive integer. Let $\left\{\mech_i\right\}_{i\in\left[m\right]}$ be $m$ mechanisms that are $\left(\gamma,\delta\right)$-untrackable for $k_i$ reports respectively. The composition of these mechanisms, $\widehat{\mech}$, is $\left(\gamma^\prime,m\delta+\delta^\prime\right)$-untrackable for
		\begin{gather*}
		\gamma^\prime=\sqrt{2m\ln\left(1/\delta^\prime\right)}\cdot\gamma + m\cdot\gamma\left(e^\gamma-1\right)
		\end{gather*}
	\end{theorem}
	The proof for this theorem, as well as the formal definition of composition, can be found in \ref{untrackable advanced composition proof}
	
	Another important question is what can be said about the untrackable guarantees in the settings where the reports are split into more than two sets, i.e.\ when we want to answer the question whether some reports were generated by a single user or any number of users. For this we define untrackable for $n$ users for $k$ reports.
	\begin{definition}[Multiple User Untrackable] \label{multiple user untrackable}
		For a domain of user data $U$, a range of reports $R$ and a report stream size $k$, and $n$ users, a report stream generator $\mech[G]^{\mech}_k$ is $\gamma$-multiple user untrackable if for all user data $u\in U$, all partitions $P=\left\{P_i\right\}_{i\in\left[n\right]}$ of $\left[k\right]$ into $n$ parts, and all output stream sets $T\subseteq R^{k}$:
		\begin{gather*}
		e^{-\gamma} \le \frac{\prod_{j \in \left[n\right]}{}\pr{\mech[G]^{\mech}_{\left|P_j\right|}\left(u\right)\in T_{P_j}}}{\pr{\mech[G]^{\mech}_k\left(u\right)\in T}}
\le e^\gamma
		\end{gather*}
	\end{definition}
We  show two connections between Definitions~\ref{approximately untrackable} and~\ref{multiple user untrackable}: the first is a general bound, essentially saying that the untrackable parameter  increases linearly in the number of users.
	\begin{theorem}\label{tracking for more than two users}
		A mechanism that is $\gamma$-untrackable for $k$ reports, is $\left(n-1\right)\gamma$-multiple user untrackable for $n$ users for $k$ reports.
	\end{theorem}
	We can significantly improve this bound for {\em permanent state mechanisms} by leveraging the fact that their untrackable parameter is linear in the number of reports used.
	\begin{theorem}\label{tracking for more than two users permanent state}
		A permanent state mechanism $\mech$, who generates reports using an $\varepsilon$-Differentially Private mechanism receiving the state as its input, is $\ceil*{\log n}\floor*{\frac{k}{2}}\varepsilon$-multiple user untrackable for $n$ users for $k$ reports.
	\end{theorem}
	The proofs of these theorems can be found in Section ~\ref{tracking for more than two users proof} and~\ref{tracking for more than two users permanent state proof}
	
	\subsection{Discussion}\label{untrackable discussion subsection}
	
	The way we defined untrackable is not the only one possible. The ``typical" attack we wish to prevent is against an adversary that sees many sets of reports and tries to identify two that belong to the same user. However, making this the basis of a definition might result in weak guarantees, as it disregards any prior information that an adversary might have. The adversary might know that Alice only lives in one of two houses, and only tries to identify where she lives. Our definition is designed to protect against exactly this kind of attacker, who only tries to distinguish whether a stream of reports was generated by Alice, or partly by Alice and partly by Bob.
	
	Another natural definition is to prevent distinguishing whether a stream of reports was generated by any combination of users vs.\ any other combination of users. Our definition, though appearing weaker than this one, actually implies it, with some deterioration to the parameter; Theorem \ref{tracking for more than two users} suggests that the parameter deteriorates linearly in the number of users, while Theorem \ref{multiple user untrackable} suggests that in some cases it can deteriorate logarithmically.
	
	Our definition also implies that it would be hard to decide whether any two reports were both generated by Alice, or one by Alice and one by Bob. This property might seem tempting as a basis of an alternative untrackable definition, but it is too weak on its own. A mechanism that has this property might have very poor protection against adversaries with access to more than two reports.
	
	Finally, Theorem \ref{undetectable from untrackable and everlasting} teaches us that our definition, when combined with everlasting privacy, naturally extends to the problem of change point (un)detection. That is, a mechanism that adheres both to the everlasting privacy requirement and our untrackable definition also protects the fact that a user changed their private value.
	
	In conclusion, This definition is strong enough to protect users against reasonable adversaries, i.e.\ ones who have some prior knowledge about the locations of users. On the other hand, while seeming weaker than other definitions it actually implies them. Additionally, as can be seen in Sections \ref{Bitwise Everlasting Privacy Mechanism section} and \ref{Report noisy inner product section}, it is achievable while also allowing for reasonable everlasting privacy guarantees and accuracy.

	\section{Mechanism Chaining} \label{mechanism chaining section}
	In this section we generalize the idea presented in Theorem \ref{untrackable bound for permanent state mechanisms} of using a Differential Privacy mechanism on the output of another such mechanism. We first provide a formal definition for this mechanism chaining, and then state and prove two theorems about the Differential Privacy guarantee achieved by doing such chaining. The first weak, but intuitive, the second much more powerful and also optimal.
	
\subsection{Definitions}
	We now present mechanism chaining in three different settings:
	In the first setting we simply define the chaining of two mechanisms as taking the output of the first and using it as the input of the second.
	\begin{definition}[$2$ Local Mechanism Chaining]
		Given two mechanisms $\mech[A]:U\rightarrow V$ and $\mech[B]:V\rightarrow O$, the chaining of these two mechanisms $\chainmech{\mech[A]}{\mech[B]}:U\rightarrow O$ is defined as $\chainmech{\mech[A]}{\mech[B]}\left(u\right)=\mech[B]\left(\mech[A]\left(u\right)\right)$
	\end{definition}
	
	The second setting we examine is the chaining of $k$ mechanism, and the third and final setting is the chaining of $k$ families of mechanisms that are not necessarily local. They are not relevant for the rest of this paper, but for completeness we present them in Appendix \ref{other trackability definitions appendix}.
	
\subsection{Differential Privacy Guarantees for Two Mechanism Chaining}
	We now present a tight bound on the Differential Privacy guarantee of the chaining of two mechanisms. We begin by presenting the ``Basic Chaining Upper Bound", which is not tight, but is perhaps more intuitive. We then present a better upper bound called the ``Advanced Chaining Upper Bound". Basic Chaining simply says that the resulting Differential Privacy is no worse than the Differential Privacy of either mechanisms.
	\begin{theorem}[Basic Chaining]\label{basic chaining}
		Given two mechanisms $\mech[A]:U\rightarrow V$ and $\mech[B]:V\rightarrow O$ that are \ldp[$\varepsilon_1$] and \ldp[$\varepsilon_2$] respectively, $\chainmech{\mech[A]}{\mech[B]}:U\rightarrow O$ is \ldp[$\min\left\{\varepsilon_1, \varepsilon_2\right\}$].
	\end{theorem}
	The advanced chaining bound is always better:
	\begin{theorem}[Advanced Chaining]\label{advanced chaining}
		Given two mechanisms $\mech[A]:U\rightarrow V$ and $\mech[B]:V\rightarrow O$ that are \ldp[$\varepsilon_1$] and \ldp[$\varepsilon_2$] respectively, $\chainmech{\mech[A]}{\mech[B]}:U\rightarrow O$ is \ldp[$\ln\frac{e^{\varepsilon_1+\varepsilon_2}+1}{e^{\varepsilon_1}+e^{\varepsilon_2}}$].
	\end{theorem}
	The proof of these theorem can be found in \ref{basic chaining proof} and \ref{advanced chaining proof}.
	The privacy parameter can be upper bounded by a more simple bound that is meaningful for small $\varepsilon_1$ and $\varepsilon_2$:
	\begin{corollary}\label{advanced chaining simpler corollary}
		Given two mechanisms $\mech[A]:U\rightarrow V$ and $\mech[B]:V\rightarrow O$ that are \ldp[$\varepsilon_1$] and \ldp[$\varepsilon_2$] respectively, $\chainmech{\mech[A]}{\mech[B]}:U\rightarrow O$ is \ldp[$\frac{1}{2}\varepsilon_1\cdot\varepsilon_2$].
	\end{corollary}
	When $\varepsilon_1$ or $\varepsilon_2$ are greater than $2$ this upper bound is worse than the bound in Theorem \ref{basic chaining}, let alone the optimal one in Theorem \ref{advanced chaining}, but otherwise this bound has little error compared to the optimal bound and is easier to work with.

	\section{(Un)Trackability in RAPPOR}\label{Untrackable in RAPPOR section}
	Equipped with a new framework to analyze tracking, we first consider one of the most significant deployments of a differentially private mechanism, used in all Chrome copies, and analyze its trackability. Introduced in~\cite{Erlingsson2014a}, RAPPOR is a DP mechanism designed to allow repeated collection of telemetry data from users in Chrome. This mechanism was the starting point of this work, since some of the goals stated in the original paper indicate the desirability of being untrackable.
	
	Roughly speaking, RAPPOR reports a value (e.g.\ the homepage of a user) from a large set. It does so with the help of a Bloom filter that initially encodes a set that contains a single element, the desired value. A Bloom filter's output is an all $0$ array that is set to $1$ at locations corresponding to hashes of the value. The mechanism proceeds to randomly flip bits in the Bloom filter, generating what we call the Permanent Randomization. At each point in time when data is to be collected, the mechanism generates a report by taking a copy of the Permanent Randomization and, again, randomly flipping bits and reporting the resulting array. The details of the mechanism can be found in the original paper, but for completeness we also present them in Appendix~\ref{RAPPOR appendix}.
	
	In the paper introducing RAPPOR, the authors mention that preventing tracking of users is an issue with their construction: \textit{``RAPPOR responses can even affect client anonymity, when they are collected on immutable client values that are the same across all clients: if the responses contain too many bits (e.g.\, the Bloom filters are too large), this can facilitate tracking clients, since the bits of the Permanent randomized responses are correlated"}. On the other hand, when talking about the reason behind the second phase of the mechanism execution, generating a report from the permanent randomization, they mention that \textit{``Instead of directly reporting $B^\prime$ [The Permanent Randomization] on every request, the client reports a randomized version of $B^\prime$. This modification significantly increases the difficulty of tracking a client based on $B^\prime$, which could otherwise be viewed as a unique identifier in longitudinal reporting scenarios"}. We wish to show that in our framework, using the same parameters they used in the RAPPOR data collections, RAPPOR is more aligned with the first statement than with the second. We analyzed RAPPOR's untrackable parameter in the worst case setting, which can be found in Appendix~\ref{worst case untracable parameter analysis appendix}. We present an analysis of the ``average case" behavior of RAPPOR.
	
\subsection{Estimated Percentile of the Trackability Random Varaible}
	We estimate the statistics of the trackabiltiy random variable for RAPPOR. In essence, the trackability random variable is the distribution of trackability leaks that happen when participating in the mechanism. The pure version of the untrackable bound in Definition \ref{approximately untrackable} is an upper bound on the possible values of the trackability random variable.

Formally, denote the RAPPOR mechanism by $\mech[R]$. For $k$ reports we define a vector of partitions $\vec{J}=\left\{J_{i}\right\}_{i\in\left[k\right]}$, where $J_{i}=\left[i\right]$. We also define two report vectors $\vec{T}=\left\{T_i\right\}_{i\in\left[k\right]}$ and $\vec{T}^\prime=\left\{T^\prime_i\right\}_{i\in\left[k\right]}$, where $T_i$ is drawn from the product distribution $\left(\mech[G]^{\mech[R]}_{i}\left(u\right),\mech[G]^{\mech[R]}_{n-i}\left(u\right)\right)$ and all of the $T^\prime_i$ are drawn from $\mech[G]^{\mech[R]}_n\left(u\right)$. The trackability random variable for $k$ reports is the value: $\tau\coloneqq\max\left\{\underset{i\in\left[\floor{\frac{k}{2}}\right]}{\max}C_{T_i,J_i},\underset{i\in\left[\floor{\frac{k}{2}}\right]}{\max}C_{T^\prime_i,J_i}\right\}$, where $C_{T,J}$ is as defined in Appendix~\ref{worst case untracable parameter analysis appendix}:
	\begin{equation*}
	C_{T,J}\coloneqq\left\lvert\ln\frac{\pr{\mech[G]^{\mech[R]}_n\left(u\right)=T}}{\pr{\mech[G]^{\mech[R]}_{\left|J\right|}\left(u\right)=T_J}\cdot\pr{\mech[G]^{\mech[R]}_{n-\left|J\right|}\left(u\right)=T_{J^\complement}}}\right\rvert
	\end{equation*}
	The random variable $\tau$ is the maximum measured tracking for the $\floor*{\frac{k}{2}}$ cases where the reports are generated by two users and the $\floor*{\frac{k}{2}}$ cases where the reports are generated by one user. In our setting a mechanism should protect against both types of cases.
		
	The measures of interest are percentiles of the trackability random variable distribution. We estimate the median and the $90^{th}$ percentile of the trackability random variable. Appendix \ref{Details of Approximation of RAPPOR} presents the details of the estimation process.

Figure \ref{rappor empirical trackability figure}, in Appendix \ref{RAPPOR appendix}, shows the estimated median and $90^{th}$ percentile of the Trackability random variable for between $2$ and $15$ reports, and their respective $95\%$ confidence interval. Our estimation shows that RAPPOR's trackability random variable's median is better than the worst case trackability, but reaches high values, around $5$ after as few as $10$ reports. The $90^{th}$ percentile is worse, reaching trackability of $5$ after as little as $7$ reports.
	
\section{Bitwise Everlasting Privacy Mechanism}\label{Bitwise Everlasting Privacy Mechanism section}

We present a mechanism for collecting statistics about the distribution of a single bit in the population, in such a way that everlasting privacy is maintained. Our mechanism is a permanent state one, using a state that consists of a noisy copy of the private bit. At each report, the user sends a noisy version of the state, effectively sending a doubly noisy version of their private bit. We show the mechanism achieves good accuracy, and reasonable everlasting privacy. Since this mechanism is a permanent state mechanism, we can use Theorem~\ref{untrackable bound for permanent state mechanisms} to give a less than reasonable upper bound on the untrackable parameter of this mechanism. We show, however a lower bound of the untrackable parameter of this mechanism that is not far off from the upper bound in Theorem~\ref{untrackable bound for permanent state mechanisms}.

Consider the mechanism where each user holds one bit, $b$. First they generate a permanent randomization, $b^\prime=b\oplus x$, where $x\sim\ber[\frac{1}{e^{\varepsilon_1}+1}]$. Then at each report they generate a report bit, $r=b^\prime\oplus y$, where $y\sim\ber[\frac{1}{e^{\varepsilon_2}+1}]$.
The aggregator receives these reports from all users and invokes the frequency oracle to output an estimate:
	\begin{equation*}
	\tilde{p}_0=
\frac{e^{\varepsilon_1+\varepsilon_2}+1-\left(e^{\varepsilon_1}+1\right)\left(e^{\varepsilon_2}+1\right)\sum_{i \in \left[n\right]} {} r_i}
{\left(e^{\varepsilon_1}-1\right)\left(e^{\varepsilon_2}-1\right)}
	\end{equation*}
	and $\tilde{p}_1=1-\tilde{p}_0$. Let $\tilde{p}$ be the vector whose coordinates are $\tilde{p}_0$ and $\tilde{p}_1$. Let $p$ be the vector of true frequencies.

\subsection{Privacy, Accuracy and Trackability}
Bitwise Everlasting Privacy is $\varepsilon_1$-EDP, outputs $\tilde{p}$ such that with probability $1-\beta$:
\begin{equation*}
\norm[\infty]{\tilde{p}-p} \le \frac{\left(\varepsilon_1+2\right)\left(\varepsilon_2+2\right)}{\varepsilon_1\cdot\varepsilon_2}\sqrt{\frac{32\ln\left(2/\beta\right)}{n}}
\end{equation*}
and is $\floor*{\frac{k}{2}}\varepsilon_2$-Untrackable, but no better than $\frac{k}{2}\varepsilon_2-\varepsilon_1-\ln 2$-Untrackable.
The proof of these claims can be found in Appendix~\ref{privacy accuracy and tracability of bitwise appendix}.
	
\section{Report Noisy Inner Product}\label{Report noisy inner product section}
	
In this section we present a method for collecting statistics about users' data when it is encoded in a vector of $d$ bits. This mechanism allows us to solve the heavy hitters  or histograms problems, while maintaining everlasting privacy. This solution achieves good accuracy with high probability and is effectively untrackable with high probability, but only for a ``not so large" number of reports (where ``not so large" is approximately the square root of the number of vectors in the state).

The ``delta" part of the untrackable bound of this solution can be small, but most likely not \emph{cryptographically} small. While in Differential Privacy one should make sure the ``delta" part is cryptographically small, it is not clear whether or not the same requirement applies to the framework of tracking.

The construction of this mechanism follows a general transformation from a Locally Differential Privacy mechanism to an Everlasting Privacy mechanism with certain trackability parameters:
memorize a fixed number ($L$) of executions of a local privacy preserving computation.
At each collection the mechanism mimics one of these stored executions, choosing one of them at random.
Everlasting Privacy is maintained by the finite access to a user's data: only $L$ total different executions are ever available to the adversary. On the other hand, in terms of trackability, as long as no two different stored execution are played, there is no difference between one user and two users. No guarantees are given if the same stored execution is chosen twice.

In our instantiation of this idea, Report Noisy Inner Product is based on creating a state that contains random $d$-bit vectors as well as their noisy inner product with the user's private value.
	
In this setting there are $n$ users. Let:
	\begin{equation*}
	\varepsilon^\prime\coloneqq\frac{\varepsilon}{2\sqrt{2L\ln\left(\frac{1}{\delta}\right)}}
	\end{equation*}

At initialization, every user $i$, with private value $u_i\in\left\{0,1\right\}^d$ chooses $L$ random vectors $\left\{v_{i,j}\right\}_{j\in\left[L\right]}$, $v_{i,j}\in\left\{0,1\right\}^d\setminus\{\boldsymbol{\vec{0}}\}$, and $L$ noisy bits $\left\{x_{i,j}\right\}_{j\in\left[L\right]}$ such that $x_{i,j}\sim\ber[\frac{1}{e^{\varepsilon^\prime}+1}]$
and calculates $b_{i,j}=\left\langle v_{i,j},\,u_i\right\rangle\oplus x_{i,j}$.

At each time of collection, every user $j$ picks at random a vector from the state generated in the previous step. That is, they choose one of the $v_i$'s generated before and the corresponding result of the inner product $s_i$. They then send it to the server. We refer to the report user $i$ sends at a given collection time as $\left(V_i, B_i\right)$ (i.e.\ the vector and the noisy inner product). The aggregator receives these reports from all users and invokes the frequency oracle to output an estimate:
\begin{equation*}
\tilde{p}_u\coloneqq\frac{2^d-1}{2^dn}\frac{e^{\varepsilon^\prime}+1}{e^{\varepsilon^\prime}-1}
\sum_{i \in \left[n\right]}{}\left(-1\right)^{\left\langle V_i,\,u\right\rangle\oplus B_i}+\frac{1}{2^d}
\end{equation*}
	
Since we never choose the vector $\vec{0}$ as one of the vectors of the state, we introduce a small bias to the probability that a report will agree with any other value than the one used to generate it. This bias is corrected by the multiplicative $\frac{2^d-1}{2^d}$ factor and the additive $\frac{1}{2^d}$ factor, resulting in an unbiased estimator.

Let $p$ be the entire true frequency vector and $\tilde{p}$ as the entire estimated frequency vector.
	
	\subsection{Privacy, Accuracy and Trackability}

The mechanism Report Noisy Inner Product (RNIP) maintains $\left(\varepsilon,\,\delta\right)$-Approximate Everlasting Privacy, outputs $\tilde{p}$ such that with probability $1-\beta$:
	\begin{align*}
	\norm[\infty]{\tilde{p}-p}
	& \le \frac{\varepsilon^\prime+2}{\varepsilon^\prime}\sqrt{\frac{8\ln(2^{d+1}/\beta)}{n}}\\
	& = O\left(\sqrt{\frac{\ln(2^{d+1}/\beta)\ln(1/\delta)L}{n\varepsilon^2}}\right)
	\end{align*}
	And it is $\left(0,\frac{k^2}{L}+\frac{L^2}{2^d}\right)$-untrackable for $k$ reports.
	
	The proofs of these claims can be found in Appendix~\ref{privacy accuracy and tracability of inner product appendix} and are similar to the analysis in~\cite{NaorPR2018}.

	\subsection{Parameter Selection}
	
When deploying this mechanism, the significant parameters considered are the everlasting privacy and desired accuracy. In our setting we have $n$ users and our data consists of values that can be encoded into $d$ bits. Assume we wish to have everlasting privacy $\left(\varepsilon,\delta\right)$ and accuracy $\alpha$ with probability $1-\beta$. By the analysis of the accuracy made in Section~\ref{report noisy inner product accuracy subsection}, the required value of the Differential Privacy parameter of every report, which we denoted $\varepsilon^\prime$, needs to be at least
$$ \frac{ 2 \sqrt{2\ln( 2^{d+1}/\beta)}}{\alpha\cdot\sqrt{n}-\sqrt{2\ln( 2^{d+1}/\beta )}}.$$
For most interesting settings we can assume that $\alpha>2\sqrt{\frac{2\ln(2^{d+1}/\beta)}{n}}$,
which allows us to choose $\varepsilon^\prime=\frac{4}{\alpha}\sqrt{\frac{2\ln(2^{d+1}/\beta)}{n}}$.
	Once we have $\varepsilon^\prime$ we can say that the mechanism needs to have a state of size at most $L=\floor*{\frac{\varepsilon^2}{8\varepsilon^{\prime 2}\ln\left(1/\delta\right)}}$. This means that the mechanism is $\left(0,\frac{k^2}{L}\right)$-Untrackable for $k$ reports.
	
To summarize, if we were to require $\left(\varepsilon,\delta\right)$-Everlasting Privacy and $\alpha$ accuracy with probability at least $1-\beta$, then for $k$ reports we can guarantee:
	\begin{equation*}
	\left(0, \widetilde{O}\left(\frac{k^2}{\alpha^2\varepsilon^2n}\right)\right)\text{-Untrackable}
	\end{equation*}
	where the $\widetilde{O}$ hides the logarithmic factors in the relaxation parameter for differential privacy $\delta$, the failure probability $\beta$ and the size of the vectors $d$.

	\section{Conclusions and Open Problems}
\label{sec:open}

The issue of using differentially private mechanisms in order to track users is a newly formulated problem. While avoiding tracking is very natural, it has not been investigated before in a formal manner. The notion of Everlasting Privacy is very tempting, and indeed, some companies implemented and deployed it. But Everlasting Privacy should be handled with caution; We have shown that one such deployment of Everlasting Privacy left much to be desired in terms of the untrackable parameter. The risks of tracking are real, and as such every mechanism deployed to a user base must try to prevent it as much as it can.

Many questions concerning tracking are open and the results presented here should be treated as a preliminary investigation. The most important one is how do you combine the constraints on accuracy and on everlasting differential privacy to produce a lower bound on the untrackable parameter. In particular, are the schemes of Sections \ref{Bitwise Everlasting Privacy Mechanism section} and \ref{Report noisy inner product section} the best one can hope for, or are there better mechanisms?
One downside of the scheme of Section \ref{Report noisy inner product section} is the rapid deterioration in the untrackable parameter once $k$ reaches $\sqrt{L}$. Is there a scheme with a more graceful degradation of the untrackable parameter?

The mechanisms we presented are permanent state mechanisms. Perhaps mechanisms which transform the state between executions can achieve better untrackable parameter bounds? Doing such a construction is delicate, since if not done correctly one of two things might happen:
\begin{enumerate}
	\item The Differential Privacy guarantee will decline the more the state alters.
	\item The accuracy will decline, as many different inputs might converge to the same states over time.
\end{enumerate}
But perhaps a clever construction of a mechanism that transforms its state can achieve a much better untrackable parameter bound for given Differential Privacy and accuracy requirements.

Also, perhaps everlasting privacy is an unreasonable demand. A mechanism that achieves privacy for many executions, but not for infinite executions, can be very suitable for practical purposes as well. If so, how can we extend these results to these ``long-lasting" privacy mechanisms?

\section*{Acknowledgements}
We thank Alexandra Korolova for many discussions and suggestions regarding the RAPPOR project as well as Guy Rothblum and Kobbi Nissim for much appreciated comments for their invaluable insights into various parts of this work. Part of this work was done while the second author was visiting the Simons Institute as part of the Data Privacy: Foundations and Applications program.
	\bibliographystyle{acm}
	\bibliography{all_refs}
	\appendix
\section{Chernoff Bounds}
	We use in some proofs the additive Chernoff bound. We use a slightly easier to work with version of Corollary A.1.7 of \cite{Alon2000}.
	\begin{theorem}\label{Chernoff Bound}
		Let $a>0$. Let $X_1,...,X_n$ be random variables, where each $X_i$ is drawn independently from a Bernoulli distribution. Denote with $X$ the average of all $X_i$. The probability the distance between $X$ and its expectation is larger than $a$ is bounded by $\pr{\abs{X-\exv[X]}>a}\le 2e^{-2n\cdot a^2}$
	\end{theorem}

	\section{Proofs for Section \ref{Stateful mechanisms and Tracking section}}
	
	\subsection{Proof for Theorem \ref{undetectable from untrackable and everlasting}}\label{proof of undetectable from untrackable and everlasting}
	\begin{proof}
		The proof follows directly from the definitions of everlasting privacy, untrackable and undetectability, by simply plugging in the bound from everlasting privacy in the bounds from the untrackable definition.
		
For the first direction notice that:
		\begin{align*}
		\pr{\mech[G]^{\mech}_k\left(u\right)\in T}
		& \le e^\gamma\cdot\pr{\mech[G]^{\mech}_{\left|J\right|}\left(u\right)\in T_J}\cdot\pr{\mech[G]^{\mech}_{k-\left|J\right|}\left(u\right)\in T_{J^\complement}}+\delta\\
		& \le e^\gamma\cdot\pr{\mech[G]^{\mech}_{\left|J\right|}\left(u\right)\in T_J}\cdot\left(e^\epsilon\pr{\mech[G]^{\mech}_{k-\left|J\right|}\left(u^\prime\right)\in T_{J^\complement}}+\delta^\prime\right)+\delta\\
		& \le e^{\gamma+\epsilon}\cdot\pr{\mech[G]^{\mech}_{\left|J\right|}\left(u\right)\in T_J}\cdot\pr{\mech[G]^{\mech}_{k-\left|J\right|}\left(u^\prime\right)\in T_{J^\complement}}+\delta+e^\gamma\delta^\prime\\
		& \le e^{\gamma+\epsilon}\cdot\pr{\mech[G]^{\mech}_{\left|J\right|}\left(u\right)\in T_J}\cdot\pr{\mech[G]^{\mech}_{k-\left|J\right|}\left(u^\prime\right)\in T_{J^\complement}}+\delta_{\max}
		\end{align*}
		Where the first inequality is due to the definition of untrackable and the second inequality is due to the definition of Everlasting Privacy. For the second direction notice that:
		\begin{align*}
		e^{-\varepsilon}\left(\pr{\mech[G]^{\mech}_{k-\left|J\right|}\left(u^\prime\right)\in T_{J^\complement}}-\delta^\prime\right)
		& \le \pr{\mech[G]^{\mech}_{k-\left|J\right|}\left(u\right)\in T_{J^\complement}}
		\end{align*}
		Which means that:
		\begin{align*}
		\pr{\mech[G]^{\mech}_{\left|J\right|}\left(u\right)\in T_J}\cdot e^{-\varepsilon}\left(\pr{\mech[G]^{\mech}_{k-\left|J\right|}\left(u^\prime\right)\in T_{J^\complement}}-\delta^\prime\right)
		& \le e^\gamma\pr{\mech[G]^{\mech}_k\left(u\right)\in T}+\delta
		\end{align*}
		And so:
		\begin{align*}
		\pr{\mech[G]^{\mech}_{\left|J\right|}\left(u\right)\in T_J}\cdot\pr{\mech[G]^{\mech}_{k-\left|J\right|}\left(u^\prime\right)\in T_{J^\complement}}
		& \le e^{\gamma+\varepsilon}\pr{\mech[G]^{\mech}_k\left(u\right)\in T}+e^\varepsilon\delta+\delta^\prime\\
		& \le e^{\gamma+\varepsilon}\pr{\mech[G]^{\mech}_k\left(u\right)\in T}+\delta_{\max}
		\end{align*}
	\end{proof}

\subsection{Proof for Theorem \ref{untrackable bound for permanent state mechanisms}}\label{proof of permanent state tracking upper bound}
\begin{proof}
		To prove this upper bound we have to prove that the two ratios in the definition of untrackable are upper bounded by $e^{\floor*{\frac{k}{2}}\varepsilon}$. Since we are talking about pure untrackable, it suffices for us to show that the bound holds for any output stream. Namely, we will consider individual output streams $t\in R^k$ as opposed to sets of report streams $T\subseteq R^k$, and show that for them the following bound holds:
		\begin{gather*}
		\frac{\pr{\mech[G]^{\mech}_{\left|J\right|}\left(u\right)=t_J}\cdot\pr{\mech[G]^{\mech}_{k-\left|J\right|}\left(u\right)=t_{J^\complement}}}{\pr{\mech[G]^{\mech}_k\left(u\right)=t}}\le e^{\floor*{\frac{k}{2}}\varepsilon}
		\end{gather*}
		and
		\begin{gather*}
		\frac{\pr{\mech[G]^{\mech}_k\left(u\right)=t}}{\pr{\mech[G]^{\mech}_{\left|J\right|}\left(u\right)=t_J}\cdot\pr{\mech[G]^{\mech}_{k-\left|J\right|}\left(u\right)=t_{J^\complement}}}\le e^{\floor*{\frac{k}{2}}\varepsilon}.
		\end{gather*}
		Since this mechanism is a Permanent State Mechanism, if its possible states are $S$, then we have that:
		\begin{gather*}
		\pr{\mech[G]^{\mech}_k\left(u\right)=t}
		=\sum_{s \in S}{}\pr{\text{State is }s}\pr{\text{Reports are }t\;|\;\text{State is } s}
		\end{gather*}
		Denote the first probability (that for a user with private data $u$, the permanent state is $s$) with  $p^{u}_{s}$ and the second probability (that a user with permanent state $s$ will output reports $t$) with $q^{s}_{t}$. The fact that we have a permanent state mechanism leads to the equality:
		\begin{equation*}
		q^{s}_{t}=q^{s}_{t_J}\cdot q^{s}_{t_{J^\complement}}
		\end{equation*}
		Let $j$ be the size of  $J$ and hence $J^\complement$ is of size $k-j$. W.l.o.g.\ assume that $j\le\frac{k}{2}$, (otherwise swap $J$ and $J^\complement$).
		
		We prove a slightly stronger bound, that if $J$ is of size $j$, then the ratios are bounded by $e^{j\varepsilon}$.
		For the first ratio, let $s^\star=\underset{s\in S}{\arg\min}\left\{q^{s}_{t_J}\right\}$. We can bound the ratio as:
		\begin{align*}
		\dfrac{\pr{\mech[G]^{\mech}_{\left|J\right|}\left(u\right)=t_J}\cdot\pr{\mech[G]^{\mech}_{k-\left|J\right|}\left(u\right)=t_{J^\complement}}}{\pr{\mech[G]^{\mech}_k\left(u\right)=t}}
		& =
		\dfrac{\left(\sigsum{s}{S}{}p^{u}_{s}q^{s}_{t_J}\right)\left(\sigsum{s}{S}{}p^{u}_{s}q^{s}_{t_{J^\complement}}\right)}{\sigsum{s}{S}{}p^{u}_{s}q^{s}_{t_J}q^{s}_{t_{J^\complement}}}\\
		& \le \dfrac{\left(\sigsum{s}{S}{}p^{u}_{s}q^{s}_{t_J}\right)\left(\sigsum{s}{S}{}p^{u}_{s}q^{s}_{t_{J^\complement}}\right)}{q^{s^\star}_{t_J}\sigsum{s}{S}{}p^{u}_{s}q^{s}_{t_{J^\complement}}}\\
		& \le \dfrac{e^{j\varepsilon}q^{s^\star}_{t_J}\left(\sigsum{s}{S}{}p^{u}_{s}\right)\left(\sigsum{s}{S}{}p^{u}_{s}q^{s}_{t_{J^\complement}}\right)}{q^{s^\star}_{t_J}\sigsum{s}{S}{}p^{u}_{s}q^{s}_{t_{J^\complement}}}\\
		& = e^{j\varepsilon}\\
		\end{align*}
		The first inequality is due to the definition of $s^\star$, the second inequality is due to the fact that the process of generating a report from a state is $\varepsilon$-DP and the report size is $j$.
		
		For the second ratio, let $s^\star=\underset{s\in S}{\arg\max}\left\{q^{s}_{t_J}\right\}$. We can bound the ratio as:
		\begin{align*}
\frac{\pr{\mech[G]^{\mech}_k(u)=t}}{\pr{\mech[G]^{\mech}_{\left|J\right|}\left(u\right)=t_J}\cdot
\pr{\mech[G]^{\mech}_{k-\left|J\right|}\left(u\right)=t_{J^\complement}}}
		& =
		\frac{\sigsum{s}{S}{}p^{u}_{s}q^{s}_{t_J}q^{s}_{t_{J^\complement}}}{\left(\sigsum{s}{S}{}p^{u}_{s}q^{s}_{t_J}\right)\left(\sigsum{s}{S}{}p^{u}_{s}p_{s,t_{J^\complement}}\right)}\\
		& \le \frac{q^{s^\star}_{t_J}\sigsum{s}{S}{}p^{u}_{s}q^{s}_{t_{J^\complement}}}{\left(\sigsum{s}{S}{}p^{u}_{s}q^{s}_{t_J}\right)\left(\sigsum{s}{S}{}p^{u}_{s}q^{s}_{t_{J^\complement}}\right)}\\
		& \le \frac{q^{s^\star}_{t_J}\sigsum{s}{S}{}p^{u}_{s}q^{s}_{t_{J^\complement}}}{e^{-j\varepsilon}q^{s^\star}_{t_J}\left(\sigsum{s}{S}{}p^{u}_{s}\right)\left(\sigsum{s}{S}{}p^{u}_{s}q^{s}_{t_{J^\complement}}\right)}\\
		& = e^{j\varepsilon}
		\end{align*}
		The first inequality is due to the definition of $s^\star$, the second inequality is due to the fact that the process of generating a report from a state is $\varepsilon$-DP and the report size is $j$.
		
		Since the maximal possible $j$ is $\floor*{\frac{k}{2}}$, we have that for all $J$ the ratios are bounded by $e^{\floor*{\frac{k}{2}}\varepsilon}$ and so the mechanism is $\floor*{\frac{k}{2}}\varepsilon$-untrackable.
	\end{proof}

	\subsection{Proof of Theorem \ref{untrackable advanced composition}}\label{untrackable advanced composition proof}
	First, we properly define the composition of $m$ mechanisms, $\left\{\mech_i\right\}_{i\in\left[m\right]}$. In our setting a user generate $m$ report streams using the $m$ mechanisms. Namely, each $\mech_i$ was used to generate $k_i$ reports. Let the sum of all $k_i$'s be $k$. Let $\widehat{\mech}$ be the composition of all the $\mech_i$'s. Formally, $\widehat{\mech}$ will first generate $k_1$ reports from the report stream generator of $\mech_1$, it will then continue to generate $k_2$ reports from the report stream generator of $\mech_2$, and so on, until all $k$ reports were generated. Let the indices of the reports generated by $\mech_i$ be $J_i$, i.e.\ $J_1=\left\{1,...,k_1\right\}$, $J_2=\left\{k_1+1,...,k_1+k_2\right\}$, and so on. Notice that the probability that $\mech[G]^{\widehat{\mech}}_k$, on input $u$, will generate a report stream $t$ of $k$ reports is exactly:
	\begin{equation*}
	\pr{\mech[G]^{\widehat{\mech}}_k\left(u\right)=t}=\pimul{i}{\left[m\right]}{}\pr{\mech[G]^{\mech_i}_{k_i}\left(u\right)=t_{J_i}}
	\end{equation*}
	since all mechanisms $\mech_i$ are executed independently.
	
	We are now ready to prove Theorem \ref{untrackable advanced composition proof}.
	\begin{proof}
		In the definition of untrackable, we consider whether a set of reports were generated by a single user or two, according to any partition. To prove that the composition mechanism is $\left(\gamma^\prime,m\delta+\delta^\prime\right)$-untrackable we will prove that the bound in the definition holds for every partition possible.
		
		Consider the partition $P=\left\{P_1,P_2\right\}$ of all the reports generated by $\widehat{\mech}$, where $P_1$ are the reports associated with the first user and $P_2$ are the reports associated with the second. We will split this partition into partitions for each mechanism separately. Namely, for every $\mech_i$ we define a partition $P^i=\left\{P_1^i,P_2^i\right\}$, such that each $P_1^i=P_1\cap J_i$ and similarly for $P_2^i$. The partition $P^i$ is exactly the partition on reports generated by $\mech_i$ induced by $P$. Notice that we allow $P_1^i$ (or $P_2^i$) to be empty.
		
		Consider new mechanisms $\left\{\mech[F]_i\right\}_{i\in\left[m\right]}$. that each receives as input a bit $b$. For every $\mech[F]_i$, If $b=0$ the mechanism outputs a stream generated by one copy of $\mech_i$, and if $b=1$ the mechanism outputs a stream generated by two independent copies of $\mech_i$ according to partition $P_i$. If either $P_1^i$ or $P_2^i$ are empty, the output $\mech[F]_i$ will not depend on its input $b$. If the $\mech_i$'s are $\left(\gamma,\delta\right)$-untrackable then the $\mech[F]_i$'s are $\left(\gamma,\delta\right)$-differentially private. This allows us to use Advanced Composition for differential privacy (Theorem \ref{advanced composition for differential privacy}) to say that the $m$-fold composition of all of the $\mech[F]_i$'s is $\left(\gamma^\prime,m\delta+\delta^\prime\right)$-differentially private for $\gamma^\prime$ as is in the theorem statement. Notice that conditioned on all mechanisms receiving input 1, the output product distribution over reports is identical to the case where all reports, for each mechanism, were generated by two users. Similarly if all inputs are 0, the output product distribution over reports is identical to the case where all reports, for each mechanism, were generated by one user. This implies that the composition of the original $\mech_i$'s, $\widehat{\mech}$, is $\left(\gamma^\prime,m\delta+\delta^\prime\right)$-untrackable.
	\end{proof}

	\subsection{Proof of Theorem \ref{tracking for more than two users}}\label{tracking for more than two users proof}
	\begin{proof}
		The proof for this theorem is very intuitive. Since the mechanism is $\gamma$-untrackable for $k$ reports, it is also $\gamma$-untrackable for fewer reports. This teaches us that by paying no more than $\gamma$ we can reduce the question of being untrackable for $n$ users to the question of being untrackable for $n-1$ users. Continuing this until we have 2 users costs us $\left(n-2\right)\gamma$ to the untrackable parameter, resulting in a total of $\left(n-1\right)\gamma$ untrackable.
		
		Formally, we prove this by induction. Assume that a mechanism $\mech$ is $\left(t-1\right)\gamma$-untrackable for $t$ users for $k$ reports. We wish to prove that the mechanism $\mech$ is $t\gamma$-untrackable for $t+1$ users for $k$ reports. The base case, $t=2$, follows directly from the fact that the mechanism is $\gamma$-untrackable for $k$ reports.
		
If the mechanism $\mech$ is $\gamma$-untrackable for $k$ reports, then it is $\gamma$-untrackable for fewer reports as well. We denote $\pr{A}\coloneqq\pr{\mech[G]^{\mech}_{\left|A\right|}\left(u\right)\in T_{A}}$. Notice that for all user data $u\in U$, all partitions $P=\left\{P_i\right\}_{i\in\left[t+1\right]}$ of $\left[k\right]$ into $t+1$ parts and all output stream sets $T\subseteq R^{k}$:
		\begin{align*}
		\pimul{j}{\left[t+1\right]}{}\pr{P_j}
		& = \pr{P_1}\cdot\pr{P_2}\cdot\pimul{j}{\left[t+1\right]\setminus\left\{1,2\right\}}{}\pr{P_j}\\
		& \le e^\gamma\pr{P_1\cup P_2}\pimul{j}{\left[t+1\right]\setminus\left\{1,2\right\}}{}\pr{\mech[G]^{\mech}_{\left|P_j\right|}\left(u\right)\in T_{P_j}}\\
		& \le e^{t\gamma}\pr{\mech[G]^{\mech}_{k}\left(u\right)\in T}
		\end{align*}
		Where the first inequality is due to the mechanism being $\gamma$-untrackable for $\left|P_1\right|+\left|P_2\right|$ reports and the second inequality is due to the induction hypothesis.
		
		Similarly, in the other direction:
		\begin{align*}
		\pr{\mech[G]^{\mech}_{k}\left(u\right)\in T}
		& \le e^{\left(t-1\right)\gamma}\pr{P_1\cup P_2}\pimul{j}{\left[t+1\right]\setminus\left\{1,2\right\}}{}\pr{P_j}\\
		& \le e^{t\gamma}\pr{P_1}\cdot\pr{P_2}\cdot\pimul{j}{\left[t+1\right]\setminus\left\{1,2\right\}}{}\pr{P_j}\\
		& = e^{t\gamma}\pimul{j}{\left[t+1\right]}{}\pr{P_j}
		\end{align*}
		Where the first inequality is due to the induction hypothesis and the second inequality is due to the mechanism being $\gamma$-untrackable for $\left|P_1\right|+\left|P_2\right|$ reports.
	\end{proof}

	\subsection{Proof of Theorem \ref{tracking for more than two users permanent state}}\label{tracking for more than two users permanent state proof}

\begin{proof}
		The proof for this is also rather intuitive. Theorem \ref{untrackable bound for permanent state mechanisms} teaches us that we can exchange the probability that two sets of reports, of size totaling $k^\prime$, originated from two users to the probability they originated from one user by paying no more than $\floor*{\frac{k^\prime}{2}}\varepsilon$ in the untrackable parameter. By combining pairs of users, we can use this to reduce the question of being untrackable for $n$ users to the question of untrackable for $\ceil*{\frac{n}{2}}$ users, by paying no more than $\floor*{\frac{k}{2}}\varepsilon$. By repeating this process $\ceil*{\log n}-1$ times we can reduce the question of untrackable for $n$ users to the question of untrackable for $2$ users, by paying no more than $\ceil*{\log n}\floor*{\frac{k}{2}}\varepsilon$.
		
		Formally, we prove this by induction. Assume that the mechanism is $\ceil*{\log t}\floor*{\frac{k}{2}}\gamma$-untrackable for $t$ users for $k$ reports. We wish to prove that the mechanism is $\ceil*{\log \left(t+1\right)}\floor*{\frac{k}{2}}\gamma$-untrackable for $t+1$ users for $k$ reports. The base case $t=2$ follows directly from the fact that the mechanism is $\gamma$-untrackable for $k$ reports.
		
		Assume $t$ is odd. The proof is very similar when it is even, but for simplicity we will only show it for odd values of $t$. We denote $\ell\coloneqq\ceil*{\log\left(t+1\right)}$ and use the same notation for $\pr{A}$ as before. Notice that for all user data $u\in U$, all partitions $P=\left\{P_i\right\}_{i\in\left[t+1\right]}$ of $\left[k\right]$ into $t+1$ parts and all output stream sets $T\subseteq R^{k}$:
		\begin{align*}
		\pimul{j}{\left[t+1\right]}{}\pr{P_j}
		& = \pimul{j}{\left[\frac{t+1}{2}\right]}{}\pr{P_{2j}}\cdot\pr{P_{2j+1}}\\
		& \le \pimul{j}{\left[\frac{t+1}{2}\right]}{}e^{\floor*{\frac{\left|P_{2j}\cup P_{2j+1}\right|}{2}}\gamma}\pr{P_{2j}\cup P_{2j+1}}\\
		& \le e^{\floor*{\frac{k}{2}}\gamma}\pimul{j}{\left[\frac{t+1}{2}\right]}{}\pr{p_{2j}\cup P_{2j+1}}\\
		& \le e^{\ell\floor*{\frac{k}{2}}\gamma}\pr{\mech[G]^{\mech}_{k}\left(u\right)\in T}
		\end{align*}
		Where the first inequality is due to Theorem~\ref{untrackable bound for permanent state mechanisms} and the third inequality is due to the induction hypothesis.
		
Similarly, for the other direction:

\begin{align*}
		\pr{\mech[G]^{\mech}_{k}\left(u\right)\in T}
		& \le e^{\left(\ell-1\right)\floor*{\frac{k}{2}}\gamma}\pimul{j}{\left[\frac{t+1}{2}\right]}{}
\pr{P_{2j}\cup P_{2j+1}}\\
		& \le e^{\ell\floor*{\frac{k}{2}}\gamma}\pimul{j}{\left[\frac{t+1}{2}\right]}{}\pr{P_{2j}}\cdot\pr{P_{2j+1}}\\
		& = e^{\ell\floor*{\frac{k}{2}}\gamma}\pimul{j}{\left[t+1\right]}{}\pr{P_{j}}
		\end{align*}

Where the first inequality is due to the induction hypothesis and the second inequality is due to Theorem~\ref{untrackable bound for permanent state mechanisms}.
	\end{proof}

\remove{
		\begin{align*}
		\pr{\mech[G]^{\mech}_{k}\left(u\right)\in T}
		& \le e^{\left(\ceil*{\log\left( t+1\right)}-1\right)\floor*{\frac{k}{2}}\gamma}\pimul{j}{\left[\frac{t+1}{2}\right]}{}\pr{\mech[G]^{\mech}_{\left|P_{2j}\right|+\left|P_{2j+1}\right|}\left(u\right)\in T_{P_{2j}\cup P_{2j}}}\\
		& \le e^{\left(\ceil*{\log \left( t+1\right)}-1\right)\floor*{\frac{k}{2}}\gamma}\pimul{j}{\left[\frac{t+1}{2}\right]}{}\left(e^{\floor*{\frac{\left|P_{2j}\right|+\left|P_{2j+1}\right|}{2}}\gamma}\pr{\mech[G]^{\mech}_{\left|P_{2j}\right|}\left(u\right)\in T_{P_{2j}}}\cdot\pr{\mech[G]^{\mech}_{\left|P_{2j+1}\right|}\left(u\right)\in T_{P_{2j+1}}}\right)\\
		& \le e^{\left(\ceil*{\log \left( t+1\right)}-1\right)\floor*{\frac{k}{2}}\gamma}e^{\floor*{\frac{k}{2}}\gamma}\pimul{j}{\left[\frac{t+1}{2}\right]}{}\left(\pr{\mech[G]^{\mech}_{\left|P_{2j}\right|}\left(u\right)\in T_{P_{2j}}}\cdot\pr{\mech[G]^{\mech}_{\left|P_{2j+1}\right|}\left(u\right)\in T_{P_{2j+1}}}\right)\\
		& = e^{\ceil*{\log \left( t+1\right)}\floor*{\frac{k}{2}}\gamma}\pimul{j}{\left[t+1\right]}{}\pr{\mech[G]^{\mech}_{\left|P_j\right|}\left(u\right)\in T_{P_j}}
		\end{align*}
}

	\section{Mechanism Chaining Generalized Definitions}\label{other trackability definitions appendix}
	\subsection{Mechanism Chaining Definition for Multiple Mechanisms}

The second setting we examine is the chaining of $k$ mechanisms. Here the definition follows the same pattern as the previous definition; chaining $k$ mechanisms is done by taking the output of each mechanism and using it as the input of the next. An equivalent way to think about it is by recursively applying the previous definition on the $k$ mechanisms: Take two mechanisms and transform them to a single mechanism by chaining them. Now repeat the process with the resulting mechanism and the next one and so forth.
	\begin{definition}[$k$ Local Mechanism Chaining]
		Given $k$ mechanisms $\mech[A]_i:V_i\rightarrow V_{i+1}$, the $k$ chaining of these mechanisms $\multichainmech{\mech[A]_i}{i\in \left[k\right]}:V_1\rightarrow V_{k+1}$ is defined as $\multichainmech{\mech[A]_i}{i\in \left[k\right]}\left(u\right)=\mech[A]_k\left(...\mech[A]_1\left(u\right)...\right)$, or recursively, as:
		\begin{equation*}
		\multichainmech{\mech[A]_i}{i\in \left[k\right]}\left(u\right)=\mech[A]_k\left(\multichainmech{\mech[A]_i}{i\in \left[k-1\right]}\left(u\right)\right)
		\end{equation*}
	\end{definition}

	\subsection{Mechanism Chaining Definition for the Non-Local Case}\label{final trackability definitions appendix}
	Before delving into the specifics of this definition we  first give an example. Assume we wanted to do a nationwide research on people's sleeping habits. One way we might implement it is by doing a hierarchy of mechanism sets. The first mechanism set is the personal one, where each individual generates a report by applying a mechanism from this set on their own data. The second mechanism set is the familial mechanism set, where each family generates a familial report by applying a mechanism from this set on the reports generated on each of its members. The third mechanism set is the cities mechanism set, where each city generates a city report by applying a mechanism from this set on the reports generated by each family in the city. We continue on like this, defining the country, continent and finally global mechanism sets. The resulting interactions between these sets and the inputs is the general definition of mechanism chaining.

A different way to think of this is by considering the set of inputs, i.e.\ the users' private information database. By partitioning of this ``input" database and a set of mechanisms, we construct a new "output" database where each entry is the output of a mechanism from this set on a section of the partition of the ``input" database (where each section of the ``input" database's partition is used once and only once in the creation of a new entry). We repeat this procedure by using the ``output" database of the previous stage as the ``input" database of the current stage. Note that in this definition we require neither the mechanisms of each set to be the same (or even to uphold the same Differential Privacy guarantee) nor that the final ``output" database is of size $1$.
	\begin{definition}[$k$ Mechanism Chaining]
		Given $k$ sets of mechanisms $M_i=\left\{\mech_{i,j}\right\}$ and a $k$ Partitions $P^{i}$ such that $\left[\abs{P^{i}}\right]=\underset{j}{\bigcup}P^{i+1}_j$, where each $\mech_{i,j}:\underset{\ell\in P^{i}_{j}}{\prod}V_{i,\ell}\mapsto V_{i+1,j}$, the $k$ chaining of these sets mechanisms $\multichainmech{M_i}{i\in \left[k\right]}:\underset{j}{\prod}V_{1,j}\mapsto \underset{j}{\prod}V_{k+1,j}$ is defined recursively as
		\begin{equation*}
		\multichainmech{M_i}{i\in \left[k\right]}\left(D\right)=\left(\mech[M]_{k,1}\left(\multichainmech{M_i}{i\in \left[k-1\right]}\left(D\right)\right),\dots,\mech[M]_{k,\abs{P^k}}\left(\multichainmech{M_i}{i\in \left[k-1\right]}\left(D\right)\right)\right)
		\end{equation*}
	\end{definition}
	\section{Proofs for Section \ref{mechanism chaining section}}
	\begin{figure}[ht]
		\caption{Illustration of the transition probabilities for going from input $u\in U$ to output $s\in O$ through $V$, $V$ being of size $k$}
		\label{mechanism chaining figure}
		\centering
		\begin{tikzpicture}[shorten >=1pt,->]
		\tikzstyle{vertex}=[circle,fill=black!25,minimum size=17pt,inner sep=0pt]
		\node[vertex] (G-u) at (1,3) {$u$};
		\node[vertex] (G-v1) at (5,5) {$v_1$};
		\node[vertex] (G-v2) at (5,4) {$v_2$};
		\node[vertex] (G-v3) at (5,1) {$v_k$};
		\node[vertex] (G-s) at (9,3) {$s$};
		\tikzstyle{vertex}=[circle,minimum size=17pt,inner sep=0pt]
		\node[vertex] (G-ddots) at (5,2.5) {$\vdots$};
		\draw (G-u) -- (G-v1) node[draw=none,fill=none,midway,above] {$p_{v_1}$};
		\draw (G-u) -- (G-v2) node[draw=none,fill=none,midway,below] {$p_{v_2}$};
		\draw (G-u) -- (G-v3) node[draw=none,fill=none,midway,below] {$p_{v_k}$};
		\draw (G-v1) -- (G-s) node[draw=none,fill=none,midway,above] {$g^{v_1}_s$};
		\draw (G-v2) -- (G-s) node[draw=none,fill=none,midway,below] {$g^{v_2}_s$};
		\draw (G-v3) -- (G-s) node[draw=none,fill=none,midway,below] {$g^{v_k}_s$};
		\end{tikzpicture}
	\end{figure}
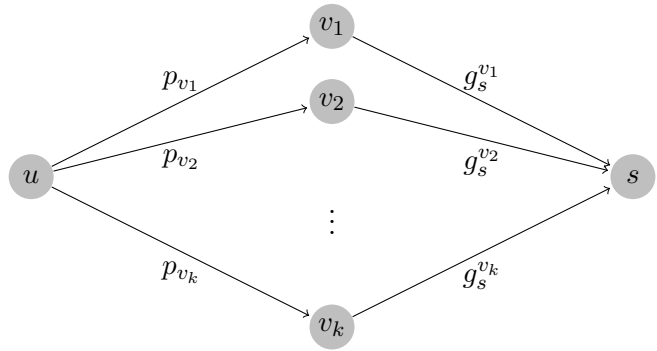
	Let $p_{v}$ be $\pr{\mech[A]\left(u\right)=v}$  and $g^{v}_{s}$ be $\pr{\mech[B]\left(v\right)=s}$. Figure~\ref{mechanism chaining figure} displays the transition probabilities for going from any input $u$ to a middle output $v$ and from $v$ to output $s$. This shows us that the probability of going from $u$ to $s$ is:
	\begin{gather*}
	\pr{\chainmech{\mech[A]}{\mech[B]}\left(u\right)=s}=\sigsum{v}{V}{}p_{v}g^{v}_{s}
	\end{gather*}
	
\subsection{Proof for Theorem \ref{basic chaining}}\label{basic chaining proof}
	\begin{proof}
The proof consists of proving that the chained mechanism is both \ldp[$\varepsilon_1$] and \ldp[$\varepsilon_2$]. The fact that the mechanism is \ldp[$\varepsilon_1$] is a direct result of $\mech[B]$ being a {\em post-processing} of the result of $\mech[A]$.
		
The proof that the mechanism is \ldp[$\varepsilon_2$] is a bit trickier. We consider the separation of the probability of sending $u$ to $s$ into the sum of the multiplications of sending $u$ to $d$ and sending $v$ to $s$ for all $v\in V$, namely:
		\begin{equation*}
		\pr{\chainmech{\mech[A]}{\mech[B]}\left(u\right)=s} = \sigsum{v}{V}{}\pr{\mech[A]\left(u\right)=v}\cdot \pr{\mech[B]\left(v\right)=s}
		\end{equation*}
		We can replace $\pr{\mech[B]\left(v\right)=s}$ with the minimal one of those $\underset{v^\prime\in V}{\min}\left\{\pr{\mech[B]\left(v^\prime\right)=s}\right\}$ by paying no more than an $e^{\varepsilon_2}$ multiplicative value, reducing all of the $\mech[B]$ contribution to the sum to a constant. Since what remains is simply a summation of the probabilities of moving from $u$ to $v$, which sums up to $1$, we can simply replace $u$ with $u^\prime$. Finally, we can replace $\underset{v^\prime\in V}{\min}\left\{\pr{\mech[B]\left(v^\prime\right)=s}\right\}$ with the corresponding original $\pr{\mech[B]\left(v\right)=s}$ without paying anything, completing the proof. Formally:
		\begin{align*}
		\pr{\chainmech{\mech[A]}{\mech[B]}\left(u\right)\in S}
		& = \sigsum[{\left(\sigsum[\pr{\mech[A]\left(u\right)=v}\cdot \pr{\mech[B]\left(v\right)=s}]{v}{V}{}\right)}]{s}{S}{}\\
		& \le \sigsum[{\left(\sigsum[\pr{\mech[A]\left(u\right)=v}\cdot e^{\varepsilon_2}\cdot \underset{v^\prime\in V}{\min}\left\{\pr{\mech[B]\left(v^\prime\right)=s}\right\}]{v}{V}{}\right)}]{s}{S}{}\\
		& = e^{\varepsilon_2}\cdot\sigsum[{\left(\sigsum[\pr{\mech[A]\left(u^\prime\right)=v}\cdot \underset{v^\prime\in V}{\min}\left\{\pr{\mech[B]\left(v^\prime\right)=s}\right\}]{v}{V}{}\right)}]{s}{S}{}\\
		& \le e^{\varepsilon_2}\cdot\sigsum[{\left(\sigsum[\pr{\mech[A]\left(u^\prime\right)=v}\cdot \pr{\mech[B]\left(v\right)=s}]{v}{V}{}\right)}]{s}{S}{}\\
		& = e^{\varepsilon_2}\pr{\chainmech{\mech[A]}{\mech[B]}\left(u^\prime\right)\in S}
		\end{align*}
		And so $\chainmech{\mech[A]}{\mech[B]}$ is \ldp[$\varepsilon_2$] and is therefore \ldp[$\min\left\{\varepsilon_1, \varepsilon_2\right\}$].
	\end{proof}

\subsection{Proof of Theorem \ref{advanced chaining}}\label{advanced chaining proof}
	
The proof is by induction over the size of the set of possible outputs of the first mechanism that the bound holds.
	First, let us prove the base case where the size of the set of all possible outputs of the first mechanism is 2. The proof begins in a general construction of such a mechanism and shows that the worst possible Differential Privacy guarantee is achieved when each mechanism is similar to a biased coin toss that has the required LDP guarantee (e.g.\, if the mechanism is {\ldp} then the worst case Differential Privacy guarantee is achieved when the mechanism outputs one output w.p. $\frac{e^\varepsilon}{e^\varepsilon + 1}$).
	\begin{lemma}\label{chaining for two}
		Given two mechanisms $\mech[A]:U\rightarrow V$ and $\mech[B]:V\rightarrow O$ that are \ldp[$\varepsilon_1$] and \ldp[$\varepsilon_2$] respectively, and given that $\left|V\right|$=2 we have that $\chainmech{\mech[A]}{\mech[B]}:U\rightarrow O$ is \ldp[$\ln\frac{e^{\varepsilon_1+\varepsilon_2}+1}{e^{\varepsilon_1}+e^{\varepsilon_2}}$].
	\end{lemma}
	\begin{proof}
		Let $V=\left\{v,w\right\}$ and let:
		\begin{align*}
		p_{v}&\coloneqq\pr{\mech[A]\left(u\right)=v}\\
		q_{v}&\coloneqq\pr{\mech[A]\left(u^\prime\right)=v}\\
		g^{v}_{s}&\coloneqq\pr{\mech[B]\left(v\right)=s}
		\end{align*}
		and similarly for $w$ (the other element in $V$). The proof first shows that the worst case differential privacy happens when $p_v=\frac{e^{\varepsilon_1}}{1+e^{\varepsilon_1}}$ and $q_v=\frac{1}{1+e^{\varepsilon_1}}$. The proof  simply uses straightforward  calculus to show that the differential privacy is the worst when both the ratios are $\abs{\log\frac{p_v}{q_v}}=e^{\varepsilon_1}$ and $\abs{\log\frac{1-p_v}{1-q_v}}=e^{\varepsilon_1}$, i.e.\ the \ldp[$\varepsilon_1$] constraint is tight everywhere. To see this notice that for all outputs $s\in O$ and pairs of user inputs $u,u^\prime\in U$ we have that:
		\begin{align*}
		\frac{\pr{\chainmech{\mech[A]}{\mech[B]}\left(u\right)=s}}{\pr{\chainmech{\mech[A]}{\mech[B]}\left(u^\prime\right)=s}} = \frac{p_{v}\cdot g^{v}_{s}+\left(1-p_{v}\right)\cdot g^{w}_{s}}{q_{v}\cdot g^{v}_{s}+\left(1-q_{v}\right)\cdot g^{w}_{s}}
		\end{align*}
		Assume w.l.o.g.\ that $g^{v}_{s}\ge g^{w}_{s}$, which means that this ratio only increases as $p_{v}$ increases and $q_{v}$ decreases. Specifically, for each value of $p_{v}$ we have that the ratio is maximized when $q_{v}$ is as small as it can get, subject to \ldp[$\varepsilon_1$]. The lower bounds on $q_{v}$ from the \ldp[$\varepsilon_1$] conditions are:
		\begin{align*}
		q_{v}\ge\frac{1}{e^{\varepsilon_1}}\cdot p_{v}\\
		\end{align*}
		And
		\begin{align*}
		1-q_{v} \le e^{\varepsilon_1}\cdot \left(1 - p_{v}\right)
		\end{align*}
		Which is equivalent to:
		\begin{align*}
		q_{v} \ge e^{\varepsilon_1}p_{v} - e^{\varepsilon_1} + 1
		\end{align*}
		And so the worst case in terms of differential privacy happens when:
		\begin{gather*}
		q_{v}=\max\left\{\dfrac{1}{e^{\varepsilon_1}}\cdot p_{v},e^{\varepsilon_1}p_{v} - e^{\varepsilon_1} + 1\right\}
		\end{gather*}
		And so the minimal value of $q_{v}$ is $\frac{1}{e^{\varepsilon_1}}\cdot p_{v}$ when:
		\begin{align*}
		\dfrac{1}{e^{\varepsilon_1}}\cdot p_{v} \ge e^{\varepsilon_1}p_{v} - e^{\varepsilon_1} + 1
		\end{align*}
		Which is equivalent to:
		\begin{gather*}
		p_{v} \le \dfrac{e^{\varepsilon_1}}{e^{\varepsilon_1} + 1}
		\end{gather*}
		On the other hand, minimal value of $q_{v}$ is $e^{\varepsilon_1}p_{v} - e^{\varepsilon_1} + 1$ when:
		\begin{gather*}
		p_{v}\ge\dfrac{e^{\varepsilon_1}}{e^{\varepsilon_1} + 1}
		\end{gather*}
		Let's split this into two cases:
		When $p_{v}\le\frac{e^{\varepsilon_1}}{e^{\varepsilon_1} + 1}$ we have that:
		\begin{align*}
		\dfrac{\pr{\chainmech{\mech[A]}{\mech[B]}\left(u\right)=s}}{\pr{\chainmech{\mech[A]}{\mech[B]}\left(u^\prime\right)=s}} &=\dfrac{p_{v}\cdot g^{v}_{s}+\left(1-p_{v}\right)\cdot g^{w}_{s}}{q_{v}\cdot g^{v}_{s}+\left(1-q_{v}\right)\cdot g^{w}_{s}}\\
		&\ge\dfrac{p_{v}\cdot g^{v}_{s}+\left(1-p_{v}\right)\cdot g^{w}_{s}}{\left(\frac{1}{e^{\varepsilon_1}}\cdot p_{v}\right)\cdot g^{v}_{s}+\left(1-\frac{1}{e^{\varepsilon_1}}\cdot p_{v}\right)\cdot g^{w}_{s}}
		\end{align*}
		Call the last expression $f_1$. If we want to minimize $f$ we have that:
		\begin{align*}
		\dfrac{\partial}{\partial p_{v}}f_1 = \dfrac{g^{w}_{s}e^{\varepsilon_1}\left(g^{v}_{s}-g^{w}_{s}\right)\left(e^{\varepsilon_1}-1\right)}{\left(\left(\frac{1}{e^{\varepsilon_1}}\cdot p_{v}\right)\cdot g^{v}_{s}+\left(1-\frac{1}{e^{\varepsilon_1}}\cdot p_{v}\right)\cdot g^{w}_{s}\right)^2}
		\end{align*}
Both the numerator and the denominator are positive (given that $g^{v}_{s}\ge g^{w}_{s}$, which we assumed w.l.o.g.), and so this function only increases as the value of $p_{v}$ increases.
		
Examining the case when $p_{v}\ge\frac{e^{\varepsilon_1}}{e^{\varepsilon_1} + 1}$ we have that:
		\begin{align*}
		\frac{\pr{\chainmech{\mech[A]}{\mech[B]}\left(u\right)=s}}{\pr{\chainmech{\mech[A]}{\mech[B]}\left(u^\prime\right)=s}} &=\dfrac{p_{v}\cdot g^{v}_{s}+\left(1-p_{v}\right)\cdot g^{w}_{s}}{q_{v}\cdot g^{v}_{s}+\left(1-q_{v}\right)\cdot g^{w}_{s}}\\
		& \le \frac{p_{v}\cdot g^{v}_{s}+\left(1-p_{v}\right)\cdot g^{w}_{s}}{\left(e^{\varepsilon_1}p_{v} - e^{\varepsilon_1} + 1\right)\cdot g^{v}_{s}+\left(e^{\varepsilon_1}-e^{\varepsilon_1}p_{v}\right)\cdot g^{w}_{s}}
		\end{align*}
		and similarly call the last expression  $f_2$. If we want to minimize $f$ we have that:
		\begin{align*}
		\frac{\partial}{\partial p_{v}} f_2 = \frac{g^{v}_{s}\left(g^{w}_{s}-g^{v}_{s}\right)\left(e^{\varepsilon_1}-1\right)}{\left(\left(e^{\varepsilon_1}p_{v} - e^{\varepsilon_1} + 1\right)\cdot g^{v}_{s}+\left(e^{\varepsilon_1}-e^{\varepsilon_1}p_{v}\right)\cdot g^{w}_{s}\right)^2}
		\end{align*}
		This is always negative (given that $g^{v}_{s}\ge g^{w}_{s}$, which we assumed w.l.o.g.), and so this function only increases as the value of $p_{v}$ decreases. So in both cases we get that the value is maximized when $p_{v}$ is on its boundary, i.e.\ $p_{v}=\frac{e^{\varepsilon_1}}{e^{\varepsilon_1}+1}$ and so $q_{v}=\frac{1}{e^{\varepsilon_1}+1}$.
		This means that:
		\begin{align*}
		\frac{\pr{\chainmech{\mech[A]}{\mech[B]}\left(u\right)=s}}{\pr{\chainmech{\mech[A]}{\mech[B]}\left(u^\prime\right)=s}} &=\frac{p_{v}\cdot g^{v}_{s}+\left(1-p_{v}\right)\cdot g^{w}_{s}}{q_{v}\cdot g^{v}_{s}+\left(1-q_{v}\right)\cdot g^{w}_{s}}\\
		&\le\frac{e^{\varepsilon_1}\cdot g^{v}_{s}+g^{w}_{s}}{g^{v}_{s}+e^{\varepsilon_1}\cdot g^{w}_{s}}
		\end{align*}
		We know that $g^{w}_{s}\ge\frac{1}{e^{\varepsilon_2}}g^{v}_{s}$ and that the ratio only increases as $g^{w}_s$ decreases, resulting in:
		
\begin{align*}
\frac{\pr{\chainmech{\mech[A]}{\mech[B]}\left(u\right)=s}}
{\pr{\chainmech{\mech[A]}{\mech[B]}\left(u^\prime\right)=s}}&\le\dfrac{e^{\varepsilon_1}\cdot g^{v}_{s}+g^{w}_{s}}{g^{v}_{s}+e^{\varepsilon_1}\cdot g^{w}_{s}}\\
		& \le \frac{e^{\varepsilon_1}\cdot g^{v}_{s}+ \frac{1}{e^{\varepsilon_2}}g^{v}_{s}}{g^{v}_{s}+e^{\varepsilon_1}\cdot\dfrac{1}{e^{\varepsilon_2}}g^{v}_{s}}\\
		&=\frac{e^{\varepsilon_1+\varepsilon_2}+1}{e^{\varepsilon_2}+e^{\varepsilon_1}}
		\end{align*}
	\end{proof}
	The proof for general sizes of $V$ relies on two facts:
	\begin{fact}\label{fraction sum lemma}
		For all $a,b,c,d>0$ we have that: $\frac{a}{b}\le\frac{c}{d}\Leftrightarrow\frac{a+c}{b+d}\ge\frac{a}{b}$
	\end{fact}
	\begin{fact}\label{fraction sub lemma}
		For all  $a,b,c,d>0$, $b>d$ we have that $\frac{a}{b}\ge\frac{c}{d}\Leftrightarrow\frac{a-c}{b-d}\ge\frac{a}{b}$
	\end{fact}
	For completeness we show their proof at Section~\ref{facts proofs appendix}.
	
	We are now ready to prove our advanced upper bound for general $V$. We use induction over the size of the set of all possible outputs of the first mechanism. For every two user inputs, the reduction will find two elements $a,b \in V$, s.t.\ if we were to output $b$ whenever we were supposed to output $a$ for both user inputs, the Differential Privacy guarantee of the chaining will only weaken. This would allow both mechanisms to maintain their original Differential Privacy guarantee and so they will have the same structure as the one in the induction hypothesis, and so the reduction holds.
	
	\begin{proof}[Proof of Theorem \ref{advanced chaining}]
		We prove this using induction on the size of $V$.
		Lemma~\ref{chaining for two} proves the base case, where $\left|V\right|=2$.
		Assume the statement is true for $\left|V\right|=k\ge2$.
		For $\left|V\right|=k+1$, using the same notation as in Lemma \ref{chaining for two} we have that $\forall s\in O$, $\forall u,u^\prime\in U$, the ratio between the probability of seeing output $s$, given that the input was $u$ and the probability of seeing output $s$, given that input was $u^\prime$ is:
		\begin{gather*}
		\dfrac{\pr{\chainmech{\mech[A]}{\mech[B]}\left(u\right)=s}}{\pr{\chainmech{\mech[A]}{\mech[B]}\left(u^\prime\right)=s}}=\dfrac{\sigsum[p_{v}\cdot g^{v}_{s}]{v}{V}{}}{\sigsum[q_{v}\cdot g^{v}_{s}]{v}{V}{}}
		\end{gather*}
		We first show that there always exist $a,b\in V$ where $a\ne b$ s.t.\ the ratio mentioned before is upper bounded by:
		\begin{gather*}
		\dfrac{\pr{\chainmech{\mech[A]}{\mech[B]}\left(u\right)=s}}{\pr{\chainmech{\mech[A]}{\mech[B]}\left(u^\prime\right)=s}}\le\dfrac{\sigsum[p_{v}\cdot g^{v}_{s}]{v}{ V\setminus\left\{a,b\right\}}{}+\left(p_{a}+p_{b}\right)\cdot g^{b}_{s}}{\sigsum[q_{v}\cdot g^{v}_{s}]{v}{ V\setminus\left\{a,b\right\}}{}+\left(q_{a}+q_{b}\right)\cdot g^{b}_{s}}
		\end{gather*}
		Then we show that the right hand side is equal to:
		\begin{gather*}
		\frac{\pr{\chainmech{\mech[\hat{A}]}{\mech[\hat{B}]}\left(u\right)=s}}{\pr{\chainmech{\mech[\hat{A}]}{\mech[\hat{B}]}\left(u^\prime\right)=s}}
		\end{gather*}
		For some $\mech[\hat{A}]$ and $\mech[\hat{B}]$, for which the induction hypothesis holds.
		If there exist $a,b\in V$ s.t.\ $g^{a}_{s}=g^{b}_{s}$ then:
		\begin{gather*}
		\frac{\pr{\chainmech{\mech[A]}{\mech[B]}\left(u\right)=s}}{\pr{\chainmech{\mech[A]}{\mech[B]}\left(u^\prime\right)=s}}=\dfrac{\sigsum[p_{v}\cdot g^{v}_{s}]{v}{ V\setminus\left\{a,b\right\}}{}+\left(p_{a}+p_{b}\right)\cdot g^{b}_{s}}{\sigsum[q_{v}\cdot g^{v}_{s}]{v}{ V\setminus\left\{a,b\right\}}{}+\left(q_{a}+q_{b}\right)\cdot g^{b}_{s}}
		\end{gather*}
		Else, we have that $\forall a,b\in V$: $g^{a}_{s}\ne g^{b}_{s}$.
		
		Since $\abs{V}\ge3$ we have that we can choose $a\in V$ s.t.\ $a\notin\left\{\underset{v\in V}{\arg\max}\left\{g^{v}_{s}\right\},\underset{v\in V}{\arg\min}\left\{g^{v}_{s}\right\}\right\}$, i.e.\ $g^{a}_{s}$ is not the maximal, nor the minimal, of all $g^{v}_{s}$. From here, we are going to construct a mechanism $\mech[\hat{B}]$ which is the same as $\mech[B]$, but without the element $a$ in its domain, and a mechanism $\mech[\hat{A}]$ that would mimic $\mech[A]$ except for where $\mech[A]$ would output $a$, $\mech[\hat{A}]$ would output $b$. $b$ will be either $\underset{v\in V}{\arg\max}\left\{g^{v}_{s}\right\}$ or $\underset{v\in V}{\arg\min}\left\{g^{v}_{s}\right\}$ such that the resulting ratio $\frac{\pr{\chainmech{\mech[\hat{A}]}{\mech[\hat{B}]}\left(u\right)=s}}{\pr{\chainmech{\mech[\hat{A}]}{\mech[\hat{B}]}\left(u^\prime\right)=s}}$ will be at least as big as $\frac{\pr{\chainmech{\mech[A]}{\mech[B]}\left(u\right)=s}}{\pr{\chainmech{\mech[A]}{\mech[B]}\left(u^\prime\right)=s}}$. Figure \ref{advanced chaining figure} demonstrates this behavior.
		\begin{figure}[ht]
			\caption{Visualization of the constructions of $\mech[\hat{A}]$ and $\mech[\hat{B}]$.}
			\label{advanced chaining figure}
			\centering
			\begin{tikzpicture}[shorten >=1pt, ->]
			\tikzstyle{vertex}=[circle,fill=black!25,minimum size=17pt,inner sep=0pt]
			\node[vertex] (G-u) at (1,4) {$u$};
			\node[vertex] (G-v1) at (5,7) {$v_1$};
			\node[vertex] (G-vb) at (5,6) {$b$};
			\node[vertex] (G-va) at (5,4) {$a$};
			\node[vertex] (G-v3) at (5,1) {$v_k$};
			\node[vertex] (G-s) at (9,4) {$s$};
			\tikzstyle{vertex}=[circle,minimum size=17pt,inner sep=0pt]
			\node[vertex] (G-ddots) at (5,3) {$\vdots$};
			\node[vertex] (G-spl) at (4,4.3) {$p_{a}$};
			\node[vertex] (G-mecha) at (2.5,1) {$\mech[\hat{A}]$};
			\node[vertex] (G-mechb) at (7.5,1) {$\mech[\hat{B}]$};
			\draw (G-u) -- (G-v1) node[draw=none,fill=none,midway,above] {$p_{v_1}$};
			\draw (G-u) -- (G-vb) node[draw=none,fill=none,midway,below] {$p_{b}$};
			\draw (G-u) to [bend right] (G-vb);
			\draw (G-u) -- (G-v3) node[draw=none,fill=none,midway,below] {$p_{v_k}$};
			\draw (G-v1) -- (G-s) node[draw=none,fill=none,midway,above] {$g^{v_1}_s$};
			\draw (G-vb) -- (G-s) node[draw=none,fill=none,midway,below] {$g^{b}_s$};
			\draw (G-v3) -- (G-s) node[draw=none,fill=none,midway,below] {$g^{v_k}_s$};
			\draw[dashed] (G-u) to (G-va);
			\draw[dashed] (G-va) to (G-s);
			\end{tikzpicture}
		\end{figure}
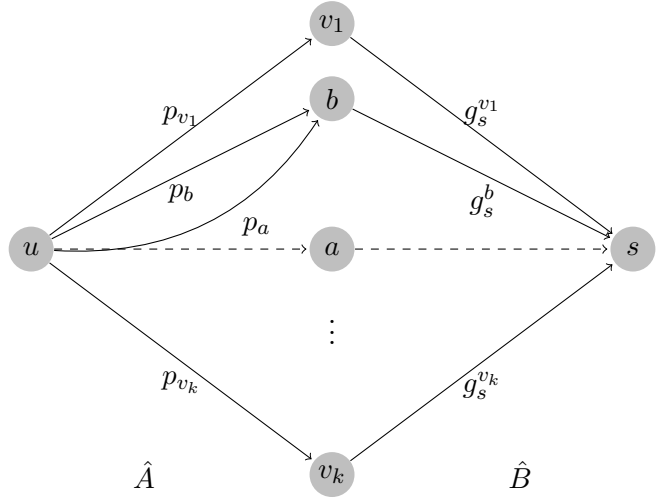
		
		The way we are going to prove this is by examining:
		\begin{equation*}
		\frac{\pr{\chainmech{\mech[A]}{\mech[B]}\left(u\right)=s}}{\pr{\chainmech{\mech[A]}{\mech[B]}\left(u^\prime\right)=s}}=\dfrac{\sigsum[p_{v}\cdot g^{v}_{s}]{v}{V}{}}{\sigsum[q_{v}\cdot g^{v}_{s}]{v}{V}{}}
		\end{equation*}
		This is a ratio of weighted sums, where each $p_{v}$ and $q_{v}$ have the same weight $g^{v}_{s}$. We use the nice property that by increasing the weight $g^{v}_{s}$ the ratio of weighted sums cannot become further from $\frac{p_{v}}{q_{v}}$, and if we were to decrease the weight $g^{v}_{s}$ it cannot become nearer. Facts \ref{fraction sum lemma} and \ref{fraction sub lemma} give us a technical way to leverage these behaviors. All in all, this allows us to exchange $g^{a}_{s}$ with $g^{b}_{s}$ for the right $b$ such that we have not pushed the ratio of the weighted sums to decrease at all.
		
		Case 1:
		\begin{equation*}
		\frac{p_{a}}{q_{a}}\ge\dfrac{\sigsum[p_{v}\cdot g^{v}_{s}]{v}{V}{}}{\sigsum[q_{v}\cdot g^{v}_{s}]{v}{V}{}}
		\end{equation*}
		In this case we wish to not push the ratio further from $\frac{p_{a}}{q_{a}}$. This means that we should choose $b=\underset{v\in V}{\arg\max}\left\{g^{v}_{s}\right\}$. Now we have that:
		$\left(g^{b}_{s}-g^{a}_{s}\right)>0$, and so we can say that $\left(g^{b}_{s}-g^{a}_{s}\right)p_{a}>0$ and $\left(g^{b}_{s}-g^{a}_{s}\right)q_{a}>0$.
		And so:
		\begin{align*}
		\frac{\left(g^{b}_{s}-g^{a}_{s}\right)p_{a}}{\left(g^{b}_{s}-g^{a}_{s}\right)q_{a}} &= \dfrac{p_{a}}{q_{a}}\\
		& \ge\dfrac{\sigsum[p_{v}\cdot g^{v}_{s}]{v}{V}{}}{\sigsum[q_{v}\cdot g^{v}_{s}]{v}{V}{}}
		\end{align*}
		Now, using Fact \ref{fraction sum lemma} (replacing parameter $a$ with $\sigsum[p_{v}\cdot g^{v}_{s}]{v}{V}{}$, parameter $b$ with $\sigsum[q_{v}\cdot g^{v}_{s}]{v}{V}{}$, parameter $c$ with $\left(g^{b}_{s}-g^{a}_{s}\right)p_{a}$ and parameter $d$ with $\left(g^{b}_{s}-g^{a}_{s}\right)q_{a}$) we have that:
		\begin{align*}
		\dfrac{\sigsum[p_{v}\cdot g^{v}_{s}]{v}{V}{}}{\sigsum[q_{v}\cdot g^{v}_{s}]{v}{V}{}} &\le\dfrac{\sigsum[p_{v}\cdot g^{v}_{s}]{v}{V}{}+\left(g^{b}_{s}-g^{a}_{s}\right)p_{a}}{\sigsum[q_{v}\cdot g^{v}_{s}]{v}{V}{}+\left(g^{b}_{s}-g^{a}_{s}\right)q_{a}}\\
		&=\dfrac{\sigsum[p_{v}\cdot g^{v}_{s}]{v}{ V\setminus\left\{a,b\right\}}{}+\left(p_{a}+p_{b}\right)\cdot g^{b}_{s}}{\sigsum[q_{v}\cdot g^{v}_{s}]{v}{ V\setminus\left\{a,b\right\}}{}+\left(q_{a}+q_{b}\right)\cdot g^{b}_{s}}
		\end{align*}
		Case 2:
		\begin{equation*}
		\dfrac{p_{a}}{q_{a}}\le\dfrac{\sigsum[p_{v}\cdot g^{v}_{s}]{v}{V}{}}{\sigsum[q_{v}\cdot g^{v}_{s}]{v}{V}{}}
		\end{equation*}
		In this case we wish to not push the ratio closer to $\frac{p_{a}}{q_{a}}$. This means that we should choose $b=\underset{v\in V}{\arg\min}\left\{g^{v}_{s}\right\}$. And so we have that:
		$\left(g^{a}_{s}-g^{b}_{s}\right)>0$, resulting in $\left(g^{a}_{s}-g^{b}_{s}\right)p_{a}>0$ and $\left(g^{a}_{s}-g^{b}_{s}\right)q_{a}>0$.
		Note also that:
		\begin{align*}
		\sigsum[q_{v}\cdot g^{v}_{s}]{v}{V}{}-\left(g^{a}_{s}-g^{b}_{s}\right)q_{a} &= \sigsum[q_{v}\cdot g^{v}_{s}]{v}{V}{}+\left(g^{b}_{s}-g^{a}_{s}\right)q_{a}\\
		&\sigsum[q_{v}\cdot g^{v}_{s}]{v}{ V\setminus\left\{a,b\right\}}{}+\left(q_{a}+q_{b}\right)\cdot g^{b}_{s}\\
		&>0
		\end{align*}
		We now use Fact \ref{fraction sub lemma} (replacing parameter $a$ with $\sigsum[p_{v}\cdot g^{v}_{s}]{v}{V}{}$, parameter $b$ with $\sigsum[q_{v}\cdot g^{v}_{s}]{v}{V}{}$, parameter $c$ with $\left(g^{a}_{s}-g^{b}_{s}\right)p_{a}$ and parameter $d$ with $\left(g^{a}_{s}-g^{b}_{s}\right)q_{a}$) to get that:
		\begin{align*}
		\frac{\sigsum[p_{v}\cdot g^{v}_{s}]{v}{V}{}}{\sigsum[q_{v}\cdot g^{v}_{s}]{v}{V}{}}&\le\dfrac{\sigsum[p_{v}\cdot g^{v}_{s}]{v}{V}{}-\left(g^{a}_{s}-g^{b}_{s}\right)p_{a}}{\sigsum[q_{v}\cdot g^{v}_{s}]{v}{V}{}-\left(g^{a}_{s}-g^{b}_{s}\right)q_{a}}\\
		&=\frac{\sigsum[p_{v}\cdot g^{v}_{s}]{v}{ V\setminus\left\{a,b\right\}}{}+\left(p_{a}+p_{b}\right)\cdot g^{b}_{s}}{\sigsum[q_{v}\cdot g^{v}_{s}]{v}{ V\setminus\left\{a,b\right\}}{}+\left(q_{a}+q_{b}\right)\cdot g^{b}_{s}}
		\end{align*}
		And so we have proven that there always exist $a,b\in V,a\ne b$ s.t.:
		\begin{gather*}
		\frac{\pr{\chainmech{\mech[A]}{\mech[B]}\left(u\right)=s}}{\pr{\chainmech{\mech[A]}{\mech[B]}\left(u^\prime\right)=s}}\le\dfrac{\sigsum[p_{v}\cdot g^{v}_{s}]{v}{ V\setminus\left\{a,b\right\}}{}+\left(p_{a}+p_{b}\right)\cdot g^{b}_{s}}{\sigsum[q_{v}\cdot g^{v}_{s}]{v}{ V\setminus\left\{a,b\right\}}{}+\left(q_{a}+q_{b}\right)\cdot g^{b}_{s}}
		\end{gather*}
		We now construct two mechanisms $\mech[\hat{A}]$ and $\mech[\hat{B}]$ that hold the induction hypothesis and have:
		\begin{equation*}
		\frac{\pr{\chainmech{\mech[\hat{A}]}{\mech[\hat{B}]}\left(u\right)=s}}{\pr{\chainmech{\mech[\hat{A}]}{\mech[\hat{B}]}\left(u^\prime\right)=s}}=\dfrac{\sigsum[p_{v}\cdot g^{v}_{s}]{v}{ V\setminus\left\{a,b\right\}}{}+\left(p_{a}+p_{b}\right)\cdot g^{b}_{s}}{\sigsum[q_{v}\cdot g^{v}_{s}]{v}{ V\setminus\left\{a,b\right\}}{}+\left(q_{a}+q_{b}\right)\cdot g^{b}_{s}}
		\end{equation*}
		
		The first mechanism is $\mech[\hat{A}]:U\rightarrow V\backslash\left\{a\right\}$ s.t.\ $\forall u\in U$ we have that:
		\[
		\pr{\mech[\hat{A}]\left(u\right)=v} =
		\begin{cases}
		p_{v} & \text{for }v\in V\backslash\left\{a,b\right\}\\
		p_{a}+p_{b} & \text{otherwise}
		\end{cases}
		\]
		And similarly for $q_v$. Notice that for all $u,u^\prime$ we have that $\forall v\in V\backslash\left\{a,b\right\}$ the ratio between the probability of having output $v$ given that the input was $u$ and the probability of having output $v$ given that the input was $u^\prime$ is upper bounded by:
		\begin{align*}
		\dfrac{\pr{\mech[\hat{A}]\left(u\right)=v}}{\pr{\mech[\hat{A}]\left(u^\prime\right)=v}} & = \dfrac{p_{v}}{q_{v}}\\
		& \le e^{\varepsilon_1}
		\end{align*}
		And
		\begin{align*}
		\dfrac{\pr{\mech[\hat{A}]\left(u\right)=b}}{\pr{\mech[\hat{A}]\left(u^\prime\right)=b}} & = \dfrac{p_{a}+p_{b}}{q_{a}+q_{a}}\\
		& \le \dfrac{e^{\varepsilon_1}q_{a}+e^{\varepsilon_1}q_{b}}{q_{a}+q_{a}}\\
		& = e^{\varepsilon_1}
		\end{align*}
		And so $\mech[\hat{A}]$ is \ldp[$\varepsilon_1$]. Notice that this probability distribution is well defined, since $\forall u\in U$ the following equality holds:
		\begin{align*}
		\sigsum[\pr{\mech[\hat{A}]\left(u\right)=v}]{v}{V\backslash\left\{a\right\}}{} & = \sigsum[\pr{\mech[\hat{A}]\left(u\right)=v}+\pr{\mech[\hat{A}]\left(u\right)=b}]{v}{V\backslash\left\{a,b\right\}}{}\\
		& = \sigsum[p_{v}]{v}{ V\setminus\left\{a,b\right\}}{}+p_{a}+p_{b}\\
		& = 1
		\end{align*}
		We define $\mech[\hat{B}]$ as $\mech[\hat{B}]:V\backslash\left\{a\right\}\rightarrow O$ s.t.\ $\forall s\in O$, $\forall v\in V\backslash\left\{a\right\}$ we have that $\pr{\mech[\hat{B}]\left(v\right)=s}=g^{v}_{s}$.
		Notice that this probability distribution is well defined, since $\forall v\in V\backslash\left\{a\right\}$ the following equation holds:
		\begin{gather*}
		\sigsum[\pr{\mech[\hat{B}]\left(v\right)=s}]{s}{O}{}=\sigsum[g^{v}_{s}]{s}{O}{}=1
		\end{gather*}
		Notice that for all $v,v^\prime$ we have that $\forall s\in O$:
		\begin{align*}
		\dfrac{\pr{\mech[\hat{B}]\left(v\right)=s}}{\pr{\mech[\hat{B}]\left(v^\prime\right)=s}} & = \dfrac{g^{v}_{s}}{g^{v^\prime}_{s}}\\
		& \le e^{\varepsilon_2}
		\end{align*}
		And so $\mech[\hat{B}]$ is \ldp[$\varepsilon_2$].
		
		We have that $\forall u,u^\prime\in U$, $s\in O$, using the chaining of mechanisms $\mech[\hat{A}]$ and $\mech[\hat{B}]$, the ratio between the probability of seeing output $s$ given that the input was $u$ and the probability of seeing output $s$ given that the input was $u^\prime$ is:
		\begin{align*}
		\dfrac{\pr{\chainmech{\mech[\hat{A}]}{\mech[\hat{B}]}\left(u\right)=s}}{\pr{\chainmech{\mech[\hat{A}]}{\mech[\hat{B}]}\left(u^\prime\right)=s}} & = \dfrac{\sigsum[\pr{\mech[\hat{A}]\left(u\right)=v}\cdot\pr{\mech[\hat{B}]\left(v\right)=s}]{v}{V\backslash\left\{a\right\}}{}}{\sigsum[\pr{\mech[\hat{A}]\left(u^\prime\right)=v}\cdot\pr{\mech[\hat{B}]\left(v\right)=s}]{v}{V\backslash\left\{a\right\}}{}}\\
		& = \dfrac{\sigsum[p_{v}\cdot g^{v}_{s}]{v}{ V\setminus\left\{a,b\right\}}{}+\left(p_{a}+p_{b}\right)\cdot g^{b}_{s}}{\sigsum[q_{v}\cdot g^{v}_{s}]{v}{ V\setminus\left\{a,b\right\}}{}+\left(q_{a}+q_{b}\right)\cdot g^{b}_{s}}
		\end{align*}
		But since $\left|V\backslash\left\{a\right\}\right|=k$ we can use the induction hypothesis to say that:
		\begin{gather*}
		\dfrac{\pr{\chainmech{\mech[\hat{A}]}{\mech[\hat{B}]}\left(u\right)=s}}{\pr{\chainmech{\mech[\hat{A}]}{\mech[\hat{B}]}\left(u^\prime\right)=s}}\le\dfrac{e^{\varepsilon_1+\varepsilon_2}+1}{e^{\varepsilon_1}+e^{\varepsilon_2}}
		\end{gather*}
		And as we have proven before, $\exists a,b\in V$ such that:
		\begin{align*}
		\dfrac{\pr{\chainmech{\mech[A]}{\mech[B]}\left(u\right)=s}}{\pr{\chainmech{\mech[A]}{\mech[B]}\left(u^\prime\right)=s}} & \le\dfrac{\sigsum[p_{v}\cdot g^{v}_{s}]{v}{ V\setminus\left\{a,b\right\}}{}+\left(p_{a}+p_{b}\right)\cdot g^{b}_{s}}{\sigsum[q_{v}\cdot g^{v}_{s}]{v}{ V\setminus\left\{a,b\right\}}{}+\left(q_{a}+q_{b}\right)\cdot g^{b}_{s}}\\
		& = \dfrac{\pr{\chainmech{\mech[\hat{A}]}{\mech[\hat{B}]}\left(u\right)=s}}{\pr{\chainmech{\mech[\hat{A}]}{\mech[\hat{B}]}\left(u^\prime\right)=s}}\\
		& \le\dfrac{e^{\varepsilon_1+\varepsilon_2}+1}{e^{\varepsilon_1}+e^{\varepsilon_2}}
		\end{align*}
		Which means that $\chainmech{\mech[A]}{\mech[B]}$ is \ldp[$\ln\frac{e^{\varepsilon_1+\varepsilon_2}+1}{e^{\varepsilon_1}+e^{\varepsilon_2}}$]
	\end{proof}
	\subsection{Proofs of Facts \ref{fraction sum lemma} and \ref{fraction sub lemma}}\label{facts proofs appendix}
	For completeness we now present the proof of the two facts used in section \ref{mechanism chaining section}.
	\begin{proof}[Proof of \ref{fraction sum lemma}]
		\begin{align*}
		\dfrac{a}{b}\le\dfrac{c}{d} & \Leftrightarrow ad\le cb\\
		& \Leftrightarrow ad+ab\le cb+ab\\
		& \Leftrightarrow a\left(b+d\right)\le b\left(a+c\right)\\
		& \Leftrightarrow \dfrac{a}{b}\le\dfrac{a+c}{b+d}
		\end{align*}
	\end{proof}
	\begin{proof}[Proof of \ref{fraction sub lemma}]
		\begin{align*}
		\dfrac{a}{b}\ge\dfrac{c}{d} & \Leftrightarrow ad\ge cb\\
		& \Leftrightarrow ad-ab\ge cb-ab\\
		& \Leftrightarrow a\left(d-b\right)\le b\left(c-a\right)\\
		& \Leftrightarrow \dfrac{a}{b}\le\dfrac{c-a}{d-b}=\dfrac{a-c}{b-d}
		\end{align*}
	\end{proof}

\section{Details of the Analyses of Section \ref{Untrackable in RAPPOR section}}\label{RAPPOR appendix}

\subsection{RAPPOR Detailed Operation}
RAPPOR is characterized by a few parameters:
	\begin{itemize}
		\item The so called Bloom filter array size, $s$.
		\item The number of hash functions in the Bloom filter, $h$.
		\item The probability of using a random bit instead of the bit from the Bloom filter in the permanent randomization, $f$.
		\item $p$ and $q$, the probabilities of outputting a bit 1 in the output given that the permanent randomized Bloom filter's bit was 0 and 1 respectively.
	\end{itemize}
	
Let $U$ be the set of all possible user values. Given $h$ hash functions $\left\{H_i\right\}_{i\in\left[h\right]}$ and a Bloom filter size $s$, $\forall v\in U$ we define that:
	\begin{equation*}
	\forall j\in\left[s\right]:\ B_j=\begin{cases}
	1 & \exists i\in\left[h\right]\text{ s.t.\ } H_i\left(v\right)=j\\
	0 & \text{otherwise}
	\end{cases}
	\end{equation*}
	The permanent randomized Bloom filter $B^\prime$ is constructed by the values of $B$ such that:
	\begin{equation*}
	\forall j\in\left[s\right]:\ B^\prime_j=\begin{cases}
	1 & \text{w.p. } \frac{1}{2}f\\
	0 & \text{w.p. } \frac{1}{2}f\\
	B_i & \text{w.p. } 1-f
	\end{cases}
	\end{equation*}
	And finally the report is constructed with probabilities depending on the values of $B^\prime$ such that:
	\begin{equation*}
	\forall j\in\left[s\right]:\ \pr{S_j=1}=\begin{cases}
	p & B^\prime_j=0\\
	q & B^\prime_j=1
	\end{cases}
	\end{equation*}
	Combining all of this we get that:
	\begin{align*}
	\pr{S_j=1|B_j=1}
	& = p\cdot\pr{B^\prime_j=0|B_j=1}+q\cdot\pr{B^\prime_j=1|B_j=1}\\
	& = \left(1-\frac{1}{2}f\right)q+\frac{1}{2}fp
	\end{align*}
	and similarly we get that:
	\begin{align*}
	\pr{S_j=1|B_j=0} & = p\cdot\pr{B^\prime_j=0|B_j=0}+q\cdot\pr{B^\prime_j=1|B_j=0}\\
	& = \frac{1}{2}fq+\left(1-\frac{1}{2}f\right)p
	\end{align*}

\subsection{Worst Case Untrackable Parameter Analysis of RAPPOR}\label{worst case untracable parameter analysis appendix}
As a warm up, we use both composition theorems to get an upper bound on the pure and approximate untrackable parameters. Using theorem \ref{untrackable bound for permanent state mechanisms}, we see that a 1-bit instance of RAPPOR, i.e.\ one that has $s=1$, is $\gamma_1$-untrackable for $k$ reports, where:
\begin{align*}
	\gamma_1=\floor*{\frac{k}{2}}\max\left\{\ln\abs{\frac{p}{q}},\ln\abs{\frac{1-p}{1-q}}\right\}
\end{align*}
An $s$-bit instance of RAPPOR can be considered as the an $s$-fold composition of the 1-bit instance. Using Basic Composition, we see that $s$-bit RAPPOR is $\gamma_s$-untrackable for $k$ reports, where:
\begin{align*}
	\gamma_s=s\cdot \floor*{\frac{k}{2}}\max\left\{\ln\abs{\frac{p}{q}},\ln\abs{\frac{1-p}{1-q}}\right\}
\end{align*}
Plugging in the parameters used in their paper, $s=128$, $p=0.5$, $q=0.75$ we get that their implementation is $88\cdot\floor*{\frac{k}{2}}$-untrackable for $k$ reports.

Alternatively, using Differential Privacy's advanced composition, we see that $s$-bit RAPPOR is $\left(\gamma^\prime_s, \delta^\prime\right)$-untrackable for $k$ reports, where:
\begin{align*}
\gamma^\prime_s=\sqrt{2s\ln\left(1/\delta^\prime\right)}\cdot\gamma_1 + s\cdot\gamma_1\left(e^{\gamma_1}-1\right)
\end{align*}
For the parameters of the paper, the $e^{\gamma_1}-1$ component of the advanced bound explodes in $k$, making resulting in a worse bound than the one from Basic Composition.

But how does this compare to the best bound we can generate for RAPPOR? We begin by showing a tight bound on the worst case untrackable parameter of RAPPOR for parameters used in the original paper. These tight bounds were computed by going over all possible values of a small set of parameters that are sufficient to find the instance with the worst trackability. The first parameter we enumerate over is the size of each report set, i.e.\ the size of the report set associated with the first user. Fixing a size for each report set, we examine each position in the original Bloom filter. We check what is the amount of $1$'s in that position in the reports associated with the first user and the second user seeking when the resulting untrackable parameter is maximized.

The calculation can be simplified since all the positions in the Bloom filter where the underlying bit is $1$ behave identically. The same can be said for positions where the underlying bit is $0$. This allows us to only consider one position of each. We calculate the contribution to the untrackable parameter made by each position in the Bloom filter separately to easily find a report which maximizes the ratio in the untrackable definition. From that report, we draw a tight untrackable bound. This allows us to find the tight bound quickly, as we only need to enumerate over five parameters, each having at most $k$ different options, where $k$ is the number of reports.
	
To better understand the procedure above, let $C_{T,J}$ be the trackability of a report set $T$ with $J$ being the indices associated with one user. Essentially, $C_{T,J}$ is the lower bound on the untrackable parameter that $T$ and $J$ induce for $\mech[R]$:
	\begin{equation*}
	C_{T,J}=\left\lvert\ln\dfrac{\pr{\mech[G]^{\mech[R]}_n\left(u\right)=T}}{\pr{\mech[G]^{\mech[R]}_{\left|J\right|}\left(u\right)=T_J}\cdot
\pr{\mech[G]^{\mech[R]}_{n-\left|J\right|}\left(u\right)=T_{J^\complement}}}\right\rvert
	\end{equation*}
	Let $i$ be the number of reports the set associated with the first user. Let $x$ and $y$ be the number of $1$'s in the first and second sets of reports. We denote $q^\star_{a,b}\coloneqq q^b\left(1-q\right)^{a-b}$, and similarly $p^\star_{a,b}$. We define $C_1\left(i,x,y\right)$ as the contribution to trackability made by positions where the underlying Bloom filter is $1$. We define $C_0$ similarly, only for position where the underlying Bloom filter is $0$:
	\begin{align*}
	C_1\left(i,x,y\right) & =\frac{\left(\frac{1}{2}fq^\star_{i,x}+\left(1-\frac{1}{2}f\right)p^\star_{i,x}\right)\left(\frac{1}{2}fq^\star_{k-i,y}+\left(1-\frac{1}{2}f\right)p^\star_{k-i,y}\right)}{\frac{1}{2}fq^\star_{k,x+y}+\left(1-\frac{1}{2}f\right)p^\star_{k,x+y}}\\
	C_0\left(i,x,y\right) & =\frac{\left(\left(1-\frac{1}{2}\right)fq^\star_{i,x}+\frac{1}{2}fp^\star_{i,x}\right)\left(\left(1-\frac{1}{2}\right)fq^\star_{k-i,y}+\frac{1}{2}fp^\star_{k-i,y}\right)}{\left(1-\frac{1}{2}\right)fq^\star_{k,x+y}+\frac{1}{2}fp^\star_{k,x+y}}
	\end{align*}
	
	By examining the trackability of reports in RAPPOR, one can see that:
	\begin{equation*}
	C_{T,J}=\left\lvert\ln\pimul{\ell}{\left[s\right]}{}C_{b_\ell}\left(i,x_\ell,y_\ell\right)\right\rvert
	\end{equation*}
	Where $i$ is the size of the reports associated with the first user, i.e.\ $\abs{J}$, $x_\ell$ is the number of ones in the $\ell$th position of all the reports in $T_J$, $y_\ell$ is the number of ones in the $\ell$th position of all the reports in $T_{J^\complement}$ and $b_\ell$ is the value of the underlying Bloom filter at position $\ell$. This means that we can calculate the trackability of report sets by considering only their relevant $i$, $x$ and $y$. Therefore in order to find a tight untrackable bound, we only need to find the values of $i$, $x$ and $y$ that maximize trackability. An immediate result of this is that the maximum trackability bound, $\gamma_k$ is:
	\begin{equation*}
	\gamma_k=\underset{i\in\left[k\right]}{\max}\underset{\substack{v,v^\prime\in\left[i\right]\\w,w^\prime\in\left[k-i\right]}}{\max}\left\lvert\ln C_{0}\left(i,v,w\right)^{s-h}C_{1}\left(i,v^\prime,w^\prime\right)^{s-h}\right\rvert
	\end{equation*}
	This can be easily computed by a simple program.
	
	In the original paper introducing RAPPOR~\cite{Erlingsson2014a}, in Section 5 - ``Experiments and Evaluation", the mechanism was deployed on Chrome users to collect Windows process names. The parameters used were $s=128$, $h=2$, $f=0.5$, $p=0.5$ and $q=0.75$. We calculated the tight trackability bound for RAPPOR with those parameters. Figure \ref{rappor trackability figure} shows us the tight untrackable parameter  for up to 15 reports. After as few as 6 reports, RAPPOR is already no better than $50$-untrackable, which should be considered high. At 15 reports, the optimal untrackable parameter is almost $300$. To conclude, RAPPOR behaves poorly in the framework of worst case untrackable parameter.
	\begin{figure}[ht]
		\caption{Trackability of RAPPOR for given number of reports}
		\label{rappor trackability figure}
		\centering
		\begin{tikzpicture}[scale=1.5]
		\begin{axis}[
		xlabel={Number of Reports},
		ylabel={Tight untrackble parameter},
		xmin=0, xmax=18,
		ymin=0, ymax=300,
		xtick={0,3,6,9,12,15},
		ytick={0,50,100,150,200,250,300},
		legend pos=north west,
		ymajorgrids=true,
		grid style=dashed,
		]
		
		\addplot[
		color=blue,
		mark=square,
		]
		coordinates {
			(2, 7.714008349934885) (3, 13.810059121593174) (4, 23.160702698499367) (5, 38.32143182990363) (6, 55.01599092734287) (7, 72.14651052510693) (8, 94.90089464514297) (9, 119.36532468090775) (10, 142.67375011125492) (11, 170.17748308615091) (12, 199.4380884218118) (13, 225.9396144617209) (14, 257.327381941984) (15, 288.6453266158751)
		};
		
		\end{axis}
		\end{tikzpicture}
	\end{figure}
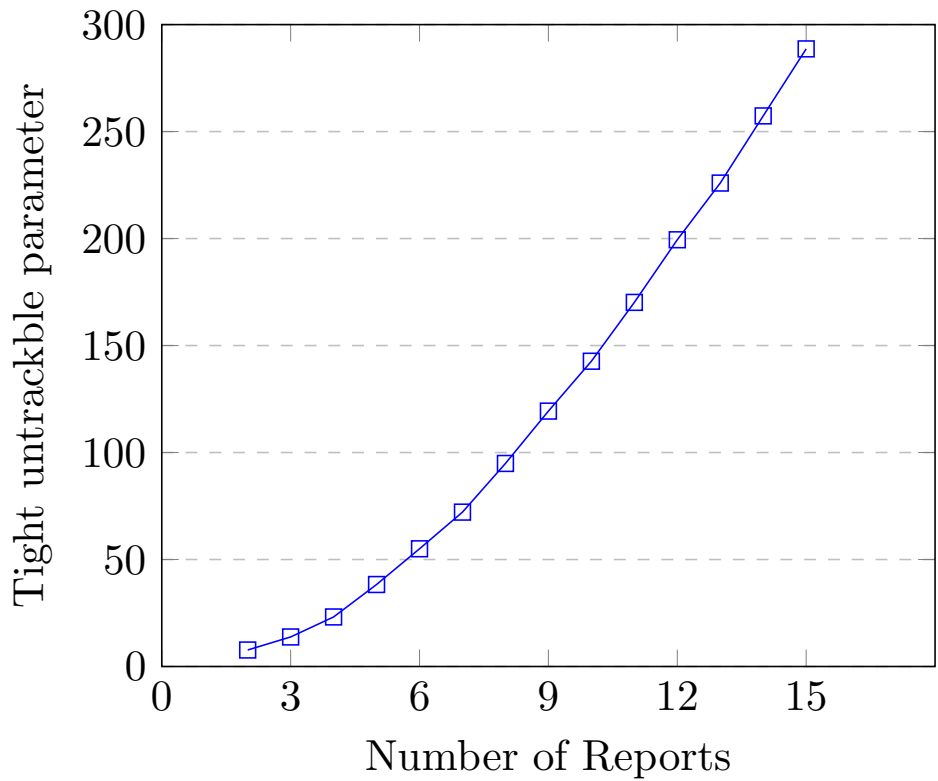
\subsection{Estimated Median and 90th Percentile Of RAPPOR Figure}
\begin{figure}[H]
	\caption{Growth of the estimated median and $90^{th}$ percentile of the trackability random variable of RAPPOR as a function of the number of reports}
	\label{rappor empirical trackability figure}
	\centering
	\begin{tikzpicture}
	\begin{axis}[
	width=13cm,
	height=19cm,
	xlabel={Number of Reports},
	ylabel={Trackability value},
	xmin=0, xmax=16,
	ymin=0, ymax=15,
	xtick={0,3,6,9,12,15},
	ytick={0,3,6,9,12,15},
	legend pos=north west,
	legend style={font=\small},
	ymajorgrids=true,
	grid style=dashed,
	legend entries={median,
		$90^{th}$ percentile
	}
	]
	
	\addplot[
	color=blue,
	mark=.,
	error bars/.cd, y dir=both, y explicit,
	]
	coordinates {(0,0)};
	\addplot[
	color=red,
	mark=.,
	error bars/.cd, y dir=both, y explicit,
	]
	coordinates {(0,0)};
	\addplot[
	only marks,
	color=blue,
	mark=.,
	error bars/.cd, y dir=both, y explicit,
	]
	coordinates {
		(2, 0.3193007504358434) +=(0, 0.010745153817822484) -=(0, 0.01231062248211856)
		(3, 0.47086424401783233) +=(0, 0.016584952895073002) -=(0, 0.01615885158821584)
		(4, 1.1042603767334072) +=(0, 0.017044448400554302) -=(0, 0.019366418493746096)
		(5, 1.4510770519450489) +=(0, 0.02205601714894101) -=(0, 0.025106895839030585)
		(6, 2.282780780058829) +=(0, 0.031770055779816175) -=(0, 0.024677574665986413)
		(7, 2.6563933536170907) +=(0, 0.03484729391232122) -=(0, 0.03821818315088876)
		(8, 3.689720500045496) +=(0, 0.04196878605807797) -=(0, 0.04061226702435761)
		(9, 4.183104590805215) +=(0, 0.046150864739274766) -=(0, 0.035151277354316335)
		(10, 5.340969664565364) +=(0, 0.039087858465450154) -=(0, 0.045457911429366504)
		(11, 5.9633857138856) +=(0, 0.04220061534783781) -=(0, 0.04622627863136586)
		(12, 7.199814161262111) +=(0, 0.04929906000029405) -=(0, 0.05668790791492029)
		(13, 7.892952626516717) +=(0, 0.06489605238448348) -=(0, 0.06258414611852459)
		(14, 9.343222828293165) +=(0, 0.05724005661886622) -=(0, 0.06333921233681394)
		(15, 10.102327281472753) +=(0, 0.06510795177882756) -=(0, 0.058399091029286865)
	};
	
	\addplot[
	only marks,
	color=red,
	mark=.,
	error bars/.cd, y dir=both, y explicit,
	]
	coordinates {
		(2, 0.9215732409547286) +=(0, 0.02056780112150136) -=(0, 0.018478454618929163)
		(3, 1.3532198978768974) +=(0, 0.02564127609497291) -=(0, 0.025174339945465363)
		(4, 2.1711208501041597) +=(0, 0.03219018680266572) -=(0, 0.03264515753460273)
		(5, 2.7427519115308883) +=(0, 0.04371164666332561) -=(0, 0.04731781590845685)
		(6, 3.8056065653744326) +=(0, 0.05157086791962229) -=(0, 0.04216381499372801)
		(7, 4.477354167400904) +=(0, 0.048495005234940436) -=(0, 0.054423526335313)
		(8, 5.722717252779148) +=(0, 0.065664535193946) -=(0, 0.06231165040958331)
		(9, 6.502512644406693) +=(0, 0.06835777324124592) -=(0, 0.06874436535747463)
		(10, 7.776868122958263) +=(0, 0.07342991566542878) -=(0, 0.08185247143353536)
		(11, 8.649237985725222) +=(0, 0.08169584824486265) -=(0, 0.07814500831671012)
		(12, 10.176025349266638) +=(0, 0.09482525994121715) -=(0, 0.09222049342997707)
		(13, 11.125639690527123) +=(0, 0.09513027662183049) -=(0, 0.09733980281157528)
		(14, 12.740815526283313) +=(0, 0.10184534432210057) -=(0, 0.10794585483290575)
		(15, 13.751703040255734) +=(0, 0.10876613815503333) -=(0, 0.1120456598587225)
	};
	\end{axis}
	\end{tikzpicture}
	\end{figure}
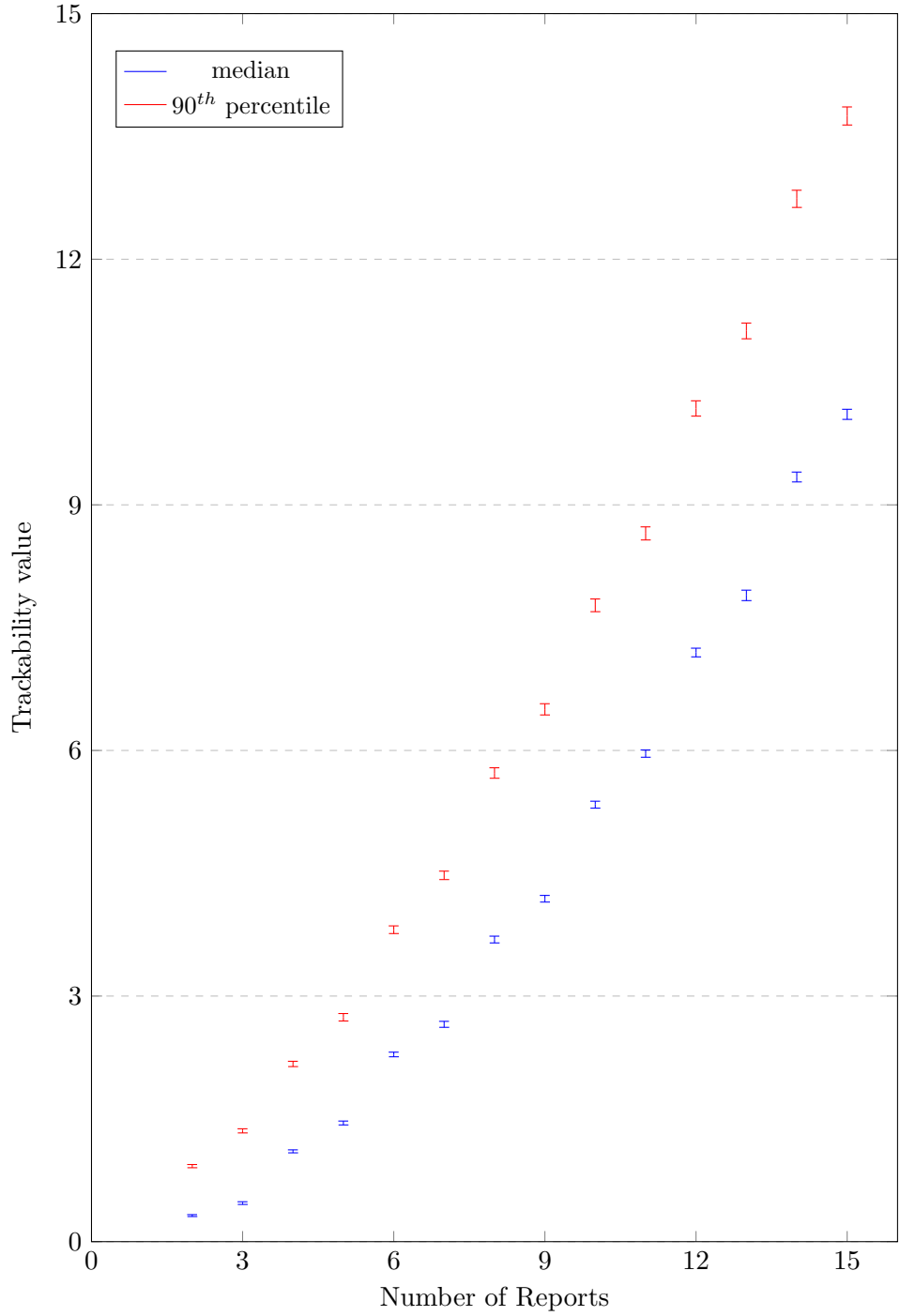
\subsection{Details of the Estimation of the Trackability Random Variable}\label{Details of Approximation of RAPPOR}
We show how to evaluate the median of the trackability random variable (the $50^{th}$ percentile) and the $90^{th}$ percentile of it. We used a $95\%$ confidence interval for the results we present in Figure \ref{rappor empirical trackability figure}.

Let $\tau$ be the random variable described above. This is the random variable whose percentiles we are interested in finding. We generated the vectors $\vec{T}$ and $\vec{T}^\prime$ according to their respective distributions. We calculate the value $\tau$ from $\vec{T}$ and $\vec{T}^\prime$ and use $\tau$ as the sample taken at each iteration. For convenience, we mark $\nu_{i}$ as the $i^{th}$ smallest sample. We generated $10000$ such samples, i.e.\ values of $\tau$, and picked $\nu_{5000}$, which is the median, namely the $50^{th}$ percentile. For the median's upper bound of the confidence interval, we picked the sample $\nu_{5100}$. This sample has the property that $51\%$ of all samples are smaller then it. Similarly, for the lower bound, we chose $\nu_{4900}$ for which $49\%$ of all samples are smaller. These bounds match a $95\%$ confidence interval for the median's estimation. We also picked $\nu_{9000}$, which is the $90^{th}$ percentile. For the $90^{th}$ percentile's upper bound of the confidence interval, we picked the sample $\nu_{9060}$. This sample has the property that $90.6\%$ of all samples are smaller then it. Similarly, for the lower bound, we chose $\nu_{8940}$ for which $89.4\%$ of all samples are smaller. These bounds match a $95\%$ confidence interval for the $90^{th}$ percentile's estimation.

Our selection of the upper and lower bounds of the confidence intervals is based on the fact that the probability that the $q^{th}$ percentile is between the $\ell^{th}$ smallest sample and the $v^{th}$ smallest sample out of $n$ samples is at least:
\begin{equation}\label{confidence equation}
B\left(v-1,n,\frac{q}{100}\right)-B\left(\ell-1,n,\frac{q}{100}\right)
\end{equation}
Where $B\left(k,n,p\right)$ is the probability of having at most $k$ $1$'s out of $n$ bits drawn from $\ber[p]$, i.e.\ the cumulative density function of the Binomial distribution. This can be seen in section 5.2.1 of the book "Statistical Intervals: a guide for Practitioners and Researchers", by Luis A. Escobar, Gerald J. Hahn and William Q. Meeker \cite{StatisticalIntervals}. To see the logic behind this calculation, first notice that:
\begin{enumerate}
	\item For the real $q^{th}$ percentile to be at most at the $v^{th}$ position in the sorted samples, less than $v$ of the samples need to be smaller.
	\item For the $q^{th}$ percentile to be at most at the $\ell^{th}$ position in the sorted samples, less than $\ell$ of the samples need to be smaller.
	\item Each sample is smaller than the $q^{th}$ percentile with probability $\frac{q}{100}$.
\end{enumerate}
This means that the probability of the first happening is $B\left(v-1,n,\frac{q}{100}\right)$ and the probability of the second happening is $B\left(\ell-1,n,\frac{q}{100}\right)$. In total, the probability that the $q^{th}$ percentile is between the $\ell^{th}$ smallest sample and the $v^{th}$ smallest sample is exactly the subtraction of the second probability, $B\left(\ell-1,n,\frac{q}{100}\right)$, from the first $B\left(v-1,n,\frac{q}{100}\right)$, resulting in Equation (\ref{confidence equation}).

For $n=10000$, plugging in $q=50$, $\ell=4900$, $v=5100$ gives the confidence level:
\begin{align*}
\pr{B\left(5099,10000,0.5\right)}-\pr{B\left(4899,10000,0.5\right)}
& = 0.95449433663\\
&\ge 0.95\\
\end{align*}
Therefore, with probability at least $0.95$, the median is between the $4900^{th}$ smallest sample and the $5100^{th}$ smallest sample.
Plugging in $q=90$, $\ell=8940$, $v=9060$ gives the confidence level:
\begin{align*}
\pr{B\left(9059,10000,0.9\right)}-\pr{B\left(8939,10000,0.9\right)}
& = 0.9545103468\\
&\ge 0.95
\end{align*}
Therefore, with probability at least $0.95$, the $90^{th}$ percentile is between the $9060^{th}$ smallest sample and the $8940^{th}$ smallest sample.

Algorithm \ref{emprirical rappor algorithm} is a pseudo code of the calculations we have performed to generate the results in Figure \ref{rappor empirical trackability figure}.\\
\begin{algorithm}[H]
	\DontPrintSemicolon
	\KwIn{RAPPOR parameters $s\in\mathbb{N}^{+}$, $h\in\left[s\right]$ and $f,p,q\in\left[0,1\right]$}
	\KwResult{Trackability random variable estimated median, $V_{50}$, and 90 percentile, $V_{90}$ and their respective $95\%$ upper confidence bound, $HighConf_{50}$ and $HighConf_{90}$ and their repsective lower confidence bound, $LowConf_{50}$, $LowConf_{90}$}
	$NSAMPS\gets 10000$\hspace{30ex}// Number of samples\\
	$Samples\gets[]$\\
	\For{$m\coloneqq1$ to $NSAMPS$} {
		$MaxSamples\gets0$\\
		\For{$i\coloneqq1$ to $\floor{\frac{k}{2}}$}{
			Choose random value $u\in U$\\
			Generate reports $T_1$ by executing $\mech[G]^{\mech[R]}_{i}\left(u\right)$\\
b			Generate reports $T_2$ by executing $\mech[G]^{\mech[R]}_{k-i}\left(u\right)$\\
			$T\gets T_1\cup T_2$\\
			$J\gets\left[i\right]$\\
			$MaxSamples\gets\max\left\{MaxSamples,C_{T,J}\right\}$
		}
		$SamplesOneUser\gets[]$\\
		\For{$i\coloneqq1$ to $\floor{\frac{k}{2}}$}{
			Choose random value $u\in U$\\
			Generate report $T$ by executing $\mech[G]^{\mech[R]}_{k}\left(u\right)$\\
			$J\gets\left[i\right]$\\
			$MaxSamples\gets\max\left\{MaxSamples,C_{T,J}\right\}$
		}
		Append $MaxSample$ to $Samples$
	}
	Sort $Samples$\\
	$V_{50}\gets Samples[0.5\cdot NSAMPS]$\\
	$LowConf_{50}\gets Samples[0.49\cdot NSAMPS]$\\
	$HighConf_{50}\gets Samples[0.51\cdot NSAMPS]$\\
	$V_{90}\gets Samples[0.9\cdot NSAMPS]$\\
	$LowConf_{90}\gets Samples[0.894\cdot NSAMPS]$\\
	$HighConf_{90}\gets Samples[0.906\cdot NSAMPS]$\\
	\caption{Estimation of the median and the $90^{th}$ percentile of RAPPOR's trackability for $k$ reports}
	\label{emprirical rappor algorithm}
\end{algorithm}
	\section{Privacy, Accuracy and Trackability of The Bitwise Mechanism}\label{privacy accuracy and tracability of bitwise appendix}
	\subsection{Privacy}
	Calculating the privacy of this mechanism is straightforward. Theorem \ref{advanced chaining} tells us that a single report is $\ln\frac{e^{\varepsilon_1+\varepsilon_2}+1}{e^{\varepsilon_1}+e^{\varepsilon_2}}$-DP. The best an adversary can do from a set of reports that originated from one user is to identify that user's state. The Construct User's Permanent State stage assures us that the state is $\varepsilon_1$-DP in the input, resulting in the mechanism being $\varepsilon_1$-EDP.
	
	\subsection{Accuracy}
	Let the group of all users who have value $0$ be $I_0$. A report $r$ generated by user $i\in I_0$ is drawn from a Bernoulli distribution: $r\sim\ber[\dfrac{e^{\varepsilon_1}+e^{\varepsilon_2}}{\left(e^{\varepsilon_1}+1\right)\left(e^{\varepsilon_2}+1\right)}]$. A report $r$ generated by user $i\notin I_0$ is drawn from a Bernoulli distribution as: $r\sim\ber[\dfrac{e^{\varepsilon_1+\varepsilon_2}+1}{\left(e^{\varepsilon_1}+1\right)\left(e^{\varepsilon_2}+1\right)}]$.
	Let $p_0$ and $p_1$ be the fraction of the users who have value $0$ and $1$ respectively.
	\begin{theorem}\label{bitwise everlasting privacy accuracy}
		Let $\beta\in\left[0,1\right]$. With probability $1-\beta$ over the randomness of all users, Bitwise Everlasting Privacy outputs $\tilde{p}_0$ such that,
		\begin{equation*}
		\abs{\tilde{p}_0-p_0}\le \dfrac{\left(e^{\varepsilon_1}+1\right)\left(e^{\varepsilon_2}+1\right)}{\left(e^{\varepsilon_1}-1\right)\left(e^{\varepsilon_2}-1\right)}\sqrt{\dfrac{2\ln\left(2/\beta\right)}{n}}
		\end{equation*}
	\end{theorem}
	If we define $p$ as the vector whose coordinates are $p_0$ and $p_1$, and similarly $\tilde{p}$, then Theorem \ref{bitwise everlasting privacy accuracy} tells us that $\norm[\infty]{\tilde{p}-p}\le \frac{\left(e^{\varepsilon_1}+1\right)\left(e^{\varepsilon_2}+1\right)}{\left(e^{\varepsilon_1}-1\right)\left(e^{\varepsilon_2}-1\right)}\sqrt{\frac{2\ln\left(2/\beta\right)}{n}}$.
	The proof follows a ``standard" Chernoff bound structure. It reduces the question of accuracy to the setting in Theorem~\ref{Chernoff Bound}.
	\begin{proof}
		Notice that $\left|I_0\right|=p_0n$. Let $X_{i}$ be the reports generated by users in $I_0$ and $Y_{i}$  the reports generated by other users.
		We examine $S=\sum_{i \in \left[n\right]}{}r_i$, which is more suitable to bound:
		\begin{align*}
		S
		& = \frac{1}{n}\sum_{i \in \left[n\right]}{}r_i\\
		& = \frac{1}{n}\left(\sum_{i \in I_0}{}X_i+\sum_{i \in \left[n\right]\setminus I_0}{}Y_i\right)
		\end{align*}
		Where, as we mentioned before, $X_i\sim\ber[\frac{e^{\varepsilon_1}+e^{\varepsilon_2}}{\left(e^{\varepsilon_1}+1\right)\left(e^{\varepsilon_2}+1\right)}]$ and $Y_i\sim\ber[\frac{e^{\varepsilon_1+\varepsilon_2}+1}{\left(e^{\varepsilon_1}+1\right)\left(e^{\varepsilon_2}+1\right)}]$. Therefore $S$ is the average of $n$ random variables that are all drawn independently from Bernoulli distributions.
		Applying Theorem~\ref{Chernoff Bound} we have that:
		\begin{align*}
		\pr{\abs{S-\exv[S]}>\sqrt{\frac{\ln\left(2/\beta\right)}{2n}}}
		& \le 2e^{-2n\left(\sqrt{\frac{\ln\left(2/\beta\right)}{2n}}\right)^2}\\
		& = \beta
		\end{align*}
		Notice that $\tilde{p}_0=\frac{e^{\varepsilon_1+\varepsilon_2}+1-S}{\left(e^{\varepsilon_1}-1\right)\left(e^{\varepsilon_2}-1\right)}$.
		
		The expectation of $\tilde{p}_u$ is:
		\begin{align*}
		\exv[\tilde{p}_0]
		& = \exv[\dfrac{e^{\varepsilon_1+\varepsilon_2}+1-S}{\left(e^{\varepsilon_1}-1\right)\left(e^{\varepsilon_2}-1\right)}]\\
		& = \dfrac{e^{\varepsilon_1+\varepsilon_2}+1-\exv[S]}{\left(e^{\varepsilon_1}-1\right)\left(e^{\varepsilon_2}-1\right)}
		\end{align*}
		And the expectation of $S$ is:
		\begin{align*}
		\exv[S]
		& = \exv[\dfrac{1}{n}\left(\sigsum{i}{I_0}{}X_i+\sigsum{i}{\left[n\right]\setminus I_0}{}Y_i\right)]\\
		& = \dfrac{e^{\varepsilon_1}+e^{\varepsilon_2}}{\left(e^{\varepsilon_1}+1\right)\left(e^{\varepsilon_2}+1\right)}p_0+\dfrac{e^{\varepsilon_1+\varepsilon_2}+1}{\left(e^{\varepsilon_1}+1\right)\left(e^{\varepsilon_2}+1\right)}\left(1-p_0\right)\\
		& = \dfrac{e^{\varepsilon_1+\varepsilon_2}+1}{\left(e^{\varepsilon_1}+1\right)\left(e^{\varepsilon_2}+1\right)}-\dfrac{\left(e^{\varepsilon_1}-1\right)\left(e^{\varepsilon_2}-1\right)}{\left(e^{\varepsilon_1}+1\right)\left(e^{\varepsilon_2}+1\right)}p_0
		\end{align*}
		Plugging this back in $\exv[\tilde{p}_0]$ we have that:
		\begin{align*}
		\exv[\tilde{p}_0]
		& = p_0
		\end{align*}
		Which means that:
		\begin{align*}
		\pr{\abs{S-\exv[S]}>\sqrt{\dfrac{\ln\left(2/\beta\right)}{2n}}}
		& = \pr{\abs{\dfrac{\left(e^{\varepsilon_1}-1\right)\left(e^{\varepsilon_2}-1\right)}{\left(e^{\varepsilon_1}+1\right)\left(e^{\varepsilon_2}+1\right)}\left(\tilde{p}_0-p_0\right)}>\sqrt{\dfrac{\ln\left(2/\beta\right)}{2n}}}\\
		& = \beta
		\end{align*}
		and we get the required bound:
		\begin{align*}
		\pr{\abs{\tilde{p}_0-p_0}>\dfrac{\left(e^{\varepsilon_1}+1\right)\left(e^{\varepsilon_2}+1\right)}{\left(e^{\varepsilon_1}-1\right)\left(e^{\varepsilon_2}-1\right)}\sqrt{\dfrac{2\ln\left(2/\beta\right)}{n}}} & \le \beta.
		\end{align*}
	\end{proof}
	We can give  a slightly more intuitive upper bound on the error:
	\begin{corollary}
		For any $\beta >0$ with probability $1-\beta$ the maximal difference between $p$ and $\tilde{p}$ is bounded by:
		\begin{equation*}
		\norm[\infty]{\tilde{p}-p}\le\dfrac{\left(\varepsilon_1+2\right)\left(\varepsilon_2+2\right)}{\varepsilon_1\cdot\varepsilon_2}\sqrt{\dfrac{32\ln\left(2/\beta\right)}{n}}
		\end{equation*}
	\end{corollary}

	\subsection{Trackability}
	Theorem~\ref{untrackable bound for permanent state mechanisms} gives us an {\em upper bound} on the untrackable parameter of the mechanism, namely that the mechanism is $\floor*{\frac{k}{2}}\varepsilon_2$-untrackable.
	
	We show a lower bound for the untrackable parameter of the mechanism: for $k$ reports the mechanism is at least $\gamma_k^\star$-untrackable for $\gamma_k^\star=\frac{k}{2}\varepsilon_2-\varepsilon_1-\ln2$. Let $p_{x,y}$ be the probability of having a stream of $y$ reports $x$ of which are 1's. We have that:
	\begin{align*}
	p_{x,y}
	& = \frac{e^{\varepsilon_1}}{1+e^{\varepsilon_1}}\cdot\left(\frac{e^{\varepsilon_2}}{1+e^{\varepsilon_2}}\right)^x\left(\frac{1}{1+e^{\varepsilon_2}}\right)^{y-x}+\frac{1}{1+e^{\varepsilon_1}}\cdot\left(\frac{1}{1+e^{\varepsilon_2}}\right)^x\left(\frac{e^{\varepsilon_2}}{1+e^{\varepsilon_2}}\right)^{y-x}\\
	& = \dfrac{e^{\varepsilon_1+x\varepsilon_2}+e^{\left(y-x\right)\varepsilon_2}}{\left(1+e^{\varepsilon_1}\right)\left(1+e^{\varepsilon_2}\right)^y}
	\end{align*}
	Now examine the set of $k$ reports where one half of them is all $1$ bits and the other half is all $0$ bits. This means that in this case  the untrackable parameter is at least:
	\begin{align*}
	\dfrac{\left(e^{\varepsilon_1}+e^{\frac{k}{2}\varepsilon_2}\right)\left(e^{\varepsilon_1+\frac{k}{2}\varepsilon_2}+1\right)}{\left(e^{\varepsilon_1}\right)\left(e^{\frac{k}{2}\varepsilon_2}+e^{\varepsilon_1+\frac{k}{2}\varepsilon_2}\right)}
	& \ge \dfrac{\left(1+e^{\frac{k}{2}\varepsilon_2}\right)\left(e^{\varepsilon_1}+e^{\varepsilon_1+\frac{k}{2}\varepsilon_2}\right)}{\left(1+e^{\varepsilon_1}\right)\left(e^{\frac{k}{2}\varepsilon_2}+e^{\varepsilon_1+\frac{k}{2}\varepsilon_2}\right)}\\
	& = \dfrac{e^{\varepsilon_1}\left(e^{\frac{k}{2}\varepsilon_2}+1\right)^2}{\left(e^{\varepsilon_1}+1\right)^2e^{\frac{k}{2}\varepsilon_2}}\\
	& \ge \dfrac{1}{4}e^{\frac{k}{2}\varepsilon_2-\varepsilon_1}
	\end{align*}
	which gives us the lower bound of $\frac{k}{2}\varepsilon_2-\varepsilon_1-\ln2$ to the mechanism's untrackable parameter.

	\section{Privacy, Accuracy and Trackability of RNIP}\label{privacy accuracy and tracability of inner product appendix}
	\subsection{Privacy}
	The set of all vectors and the results of the inner products is essentially a composition of differential private mechanisms, where each mechanism is the noisy inner product with a single vector. Our privacy analysis is for approximate differential privacy. A pure differential privacy guarantee can also be calculated. However, in this case there is much to be gained from the fact that, under approximate differential privacy, the guarantees deteriorate only as a square root of the number of mechanisms composed. Assume we want approximate everlasting privacy $\left(\varepsilon,\,\delta\right)$ (for $\varepsilon<1$), if we define:
	\begin{equation*}
	\varepsilon^\prime\coloneqq\dfrac{\varepsilon}{2\sqrt{2L\ln\left(\frac{1}{\delta}\right)}}
	\end{equation*}
	We choose $\varepsilon^\prime$ as the differential privacy guarantee of each noisy inner product. We can use Advanced Composition \cite{DworkRV2010} to show that the state, consisting of the $L$ vectors and the corresponding inner products, is $\left(\varepsilon, \delta\right)$-Differentially Private in the user's data. Consider an adversary who is only given access to the reports of users. The best this adversary can do is to restore the entire state of the user, resulting in the mechanism being $\left(\varepsilon, \delta\right)$-Differentially Private.
	
	\subsection{Accuracy}\label{report noisy inner product accuracy subsection}
	
For the accuracy analysis, we first present an accurate bound that is less intuitive. This analysis is similar to the one done in the ``Practical Locally Private Heavy Hitters" paper~\cite{BassilyNST2017} discussed in~\ref{background subsection}. It is also very similar to the one done in the ``How to (not) share a password: Privacy preserving protocols for finding heavy hitters with adversarial behavior" paper~\cite{NaorPR2018}, described as well in~\ref{background subsection}, which uses a similar mechanism, only with a single vector per user. The proof uses Theorem~\ref{Chernoff Bound} to bound the probability that a single frequency is estimated poorly, and by using union bound, we bound the probability that any frequency are estimated poorly. After we prove this, we present a more intuitive, easier to work with bound.
	\begin{theorem}
		Let $\beta\in\left[0,1\right]$.
		With probability $1-\beta$ over the randomness of all users, Report noisy inner product outputs $\tilde{p}$ such that,
		\begin{equation*}
		\norm[\infty]{\tilde{p}-p}\le \dfrac{2^d-1}{2^{d-1}}\dfrac{e^{\varepsilon^\prime}+1}{e^{\varepsilon^\prime}-1}\sqrt{\dfrac{2\ln\frac{2^{d+1}}{\beta}}{n}}
		\end{equation*}
	\end{theorem}
	The proof follows a standard Chernoff bound structure. It begins by bounding the error of the estimate of every possible user data $u$ using the additive Chernoff bound, and with union bound, shows that the bound holds for all possible user data with probability $1-\beta$.
	\begin{proof}
		We begin with a lemma that bounds the probability of having a big error on the estimation of the frequency of a single private value.
		\begin{lemma}\label{noisy inner product single value bound}
			Let $\beta\in\left[0,1\right]$. For all values $u \in \{0,1\}^d \backslash \{0\}$ we have that,
			\begin{equation*}
			\pr{\lvert\tilde{p}_u-p_u\rvert>\dfrac{2^d-1}{2^{d-1}}\dfrac{e^{\varepsilon^\prime}+1}{e^{\varepsilon^\prime}-1}\sqrt{\dfrac{2\ln\frac{2^{d+1}}{\beta}}{n}}}\le\dfrac{\beta}{2^d}
			\end{equation*}
			where the probability is taken over the randomness of all the users.
		\end{lemma}
		\begin{proof}
			Let $I$ be the set of all users who have value $u$. Notice that $\left|I\right|=p_un$. To fit the setting of Theorem \ref{Chernoff Bound}, we examine $g_u$, the fraction of reports that agree with $u$. In the following sum every element will be in $\mathbb{F}_2$. We will abuse notation and denote the sum of them as the summation of these elements {\em over the Reals}, i.e.\ as if every one of them was either $0$ or $1$.
			\begin{align*}
			g_u
			& \coloneqq \dfrac{1}{n}\sigsum{i}{\left[n\right]}{}\left(1\oplus\left\langle V_i,\,u\right\rangle\oplus B_i\right)\\
			& = \dfrac{1}{n}\left(\sigsum{i}{I}{}\left(1\oplus\left\langle V_i,\,u\right\rangle\oplus\left\langle V_i,\,u\right\rangle\oplus X_i\right)+\sigsum{i}{\left[n\right]\setminus I}{}\left(1\oplus\left\langle V_i,\,u\right\rangle\oplus\left\langle V_i,\,u^\prime\right\rangle\oplus X_i\right)\right)\\
			& = \dfrac{1}{n}\left(\sigsum{i}{I}{}\left(1\oplus X_i\right)+\sigsum{i}{\left[n\right]\setminus I}{}Y_i\right)
			\end{align*}
			Where $u^\prime$ is any value that is not $u$ and $Y_i$ is the XOR of $X_i$ and $\left\langle V_i,\,u\oplus u^\prime\right\rangle$. By definition, $X_i\sim\ber[\frac{1}{1+e^\varepsilon_1}]$. It is easy to see that $\left\langle V_i,\,u\oplus u^\prime\right\rangle\sim\ber[\frac{2^{d-1}}{2^d-1}]$. Since $Y_i$ is the XOR of two Bernoulli random variables, we can use the following simple fact to calculate its distribution:
			\begin{fact}
				Let $0 \leq q_1, q_2 \leq 1$. The XOR of two random variables $s\sim\ber[q_1]$ and $t\sim\ber[q_2]$ is
				\begin{equation*}
				s\oplus t\sim\ber[2q_1 q_2 -q_1-q_2+1].
				\end{equation*}
			\end{fact}
			Applying this Fact to $Y_i$ we can see that its probability distribution is
			\begin{equation*}
			Y_i\sim\ber[\dfrac{2^{d-1}\left(e^{\varepsilon^\prime}+1\right)-e^{\varepsilon^\prime}}{\left(2^d-1\right)\left(e^{\varepsilon^\prime}+1\right)}]
			\end{equation*}
			Therefore, $g_u$ is the average of $n$ random variables that are all drawn independently from Bernoulli distributions.
			Applying Theorem \ref{Chernoff Bound} we have that:
			\begin{equation}\label{chernoff of noisy innter product equation}
			\pr{\big\lvert g_u-\exv[g_u]\big\rvert>\sqrt{\dfrac{\ln\frac{2^{d+1}}{\beta}}{2n}}} \le \dfrac{\beta}{2^d}
			\end{equation}
			Notice that:
			\begin{align*}
			g_u
			& = \dfrac{1}{n}\sigsum{i}{\left[n\right]}{}\left(1\oplus\left\langle V_i,\,u\right\rangle\oplus B_i\right)\\
			& = \dfrac{1}{n}\sigsum{i}{\left[n\right]}{}\dfrac{\left(-1\right)^{\left\langle V_i,\,u\right\rangle\oplus B_i}+1}{2}
			\end{align*}
			Which means that $\tilde{p}_u=\frac{2^d-1}{2^d}\frac{e^{\varepsilon^\prime}+1}{e^{\varepsilon^\prime}-1}\left(2g_u-1\right)+\frac{1}{2^d}$.
			
			The expectation of $\tilde{p}_u$ is:
			\begin{align*}
			\exv[\tilde{p}_u]
			& =  \dfrac{2^d-1}{2^d}\dfrac{e^{\varepsilon^\prime}+1}{e^{\varepsilon^\prime}-1}\left(2\exv[g_u]-1\right)+\dfrac{1}{2^d}
			\end{align*}
			And the expectation of $g_u$ is:
			\begin{align*}
			\exv[g_u]
			& = \exv[\dfrac{1}{n}\left(\sigsum{i}{I}{}1\oplus X_i+\sigsum{i}{\left[n\right]\setminus I}{}Y_i\right)]\\
			& = \left(\dfrac{2^{d-1}\left(e^{\varepsilon^\prime}-1\right)}{\left(2^d-1\right)\left(e^{\varepsilon^\prime}+1\right)}\right)p_u+\dfrac{2^{d-1}\left(e^{\varepsilon^\prime}+1\right)-e^{\varepsilon^\prime}}{\left(2^d-1\right)\left(e^{\varepsilon^\prime}+1\right)}
			\end{align*}
			Plugging this back in $\exv[\tilde{p}_u]$ we have that $\exv[\tilde{p}_u] = p_u$.
			This allows us to take Inequality (\ref{chernoff of noisy innter product equation}) and transform it to get the bound required:
			\begin{align*}
			\pr{\big\lvert\tilde{p}_u-p_u\big\rvert>\dfrac{2^d-1}{2^{d-1}}\dfrac{e^{\varepsilon^\prime}+1}{e^{\varepsilon^\prime}-1}\sqrt{\dfrac{2\ln\frac{2^{d+1}}{\beta}}{n}}} & \le \dfrac{\beta}{2^d}
			\end{align*}
		\end{proof}
		Now using Lemma \ref{noisy inner product single value bound} and union bound on all possible user values $u$ we have that:
		\begin{align*}
		\pr{\norm[\infty]{\tilde{p}-p}\le\dfrac{2^d-1}{2^{d-1}}\dfrac{e^{\varepsilon^\prime}+1}{e^{\varepsilon^\prime}-1}\sqrt{\dfrac{2\ln\frac{2^{d+1}}{\beta}}{n}}}\ge 1-\beta
		\end{align*}
		as required.
	\end{proof}
	One can upper bound this error with a slightly more intuitive bound
	\begin{corollary}
		Let $\beta\in\left[0,1\right]$. With probability $1-\beta$ the maximal difference between $\tilde{p}$ and $p$ is bounded by:
		\begin{equation*}
		\norm[\infty]{\tilde{p}-p}\le\dfrac{\varepsilon^\prime+2}{\varepsilon^\prime}\sqrt{\dfrac{8\ln\frac{2^{d+1}}{\beta}}{n}}
		\end{equation*}
	\end{corollary}
	Notice that the sum of the coordinates of $\tilde{p}$ is $1$:
	\begin{lemma}
		The sum of all estimated frequencies sums up to one, i.e.\ $\sigsum{u}{}{}\tilde{p}_u=1$
	\end{lemma}
	\begin{proof}
		By the definition of $\tilde{p}_u$ we have that:
		\begin{align*}
		\sigsum{u}{}{}\tilde{p}_u
		& = 	\sigsum{u}{}{}\left(\dfrac{2^d-1}{2^dn}\dfrac{e^{\varepsilon^\prime}+1}{e^{\varepsilon^\prime}-1}\sigsum{i}{\left[n\right]}{}\left(-1\right)^{\left\langle V_i,\,u\right\rangle\oplus B_i}+\dfrac{1}{2^d}\right)\\
		& = 	\dfrac{2^d-1}{2^dn}\dfrac{e^{\varepsilon^\prime}+1}{e^{\varepsilon^\prime}-1}\sigsum{i}{\left[n\right]}{}\sigsum{u}{}{}\left(-1\right)^{\left\langle V_i,\,u\right\rangle}+1\\
		& = 	\dfrac{2^d-1}{2^dn}\dfrac{e^{\varepsilon^\prime}+1}{e^{\varepsilon^\prime}-1}\sigsum{i}{\left[n\right]}{}0+1\\
		& = 1
		\end{align*}
		where the second equality is due to the fact that since $V_i\ne\boldsymbol{\vec{0}}$ there exists a vector $u^\prime$ such that $\left\langle V_i,\,u^\prime\right\rangle=B_i$, and third equality is due to the fact that since $V_i\ne\boldsymbol{\vec{0}}$ there exists a vector $u^\star$ such that for all $u$, $\left\langle V_i,\,u\right\rangle=1\oplus\left\langle V_i,\,u\oplus u^\star\right\rangle$, meaning $\left(-1\right)^{\left\langle V_i,\,u\right\rangle}+\left(-1\right)^{\left\langle V_i,\,u\oplus u^\star\right\rangle}=0$.
	\end{proof}
	\subsection{Trackability}
	Our scheme is untrackable for ``not so large" number of reports. We prove that, when $k\ll\sqrt{L}$  (and $L\ll2^{\frac{d}{2}}$) then  Report Noisy Inner Product has fair chance of remaining perfectly untrackable (but this probability is {\em not} negligible):
	\begin{theorem}
		Let $k,L,d$ be positive integers. Report Noisy Inner Product, with state size $L$ for a domain of size $2^d$, is $\left(0,\frac{k^2}{L}+\frac{L^2}{2^d}\right)$-untrackable for $k$ reports.
	\end{theorem}
	\begin{proof}
		The strategy of proving this theorem is:
		\begin{enumerate}
			\item Define an event $B$ so that conditioned on it not occurring, both the single user case and the two users case are identically distributed.
			\item $B$ has a probability of at most $\frac{k^2}{2L}+\frac{L^2}{2^d}$ of occurring both in the single user case and the two users case.
		\end{enumerate}
		The event $B$ occurs when any two vectors used for two of the $k$ reports are identical. To see this, first consider the stateless version of the Report Noisy Inner Product. At each report, this mechanism chooses a random vector and reports the vector and the noisy inner product of that vector with the private value. Conditioned on $B$ not happening, the Report Noisy Inner Product is equivalent to its stateless version. Since for the stateless mechanism each report is completely independent it does not matter if the reports were generated by a single user or two. This means that, conditioned on $B$ not happening, both cases are identically distributed.
		
		For the single user case, event $B$ happens with probability at most $\binom{k}{2}/{L} \le\frac{k^2}{2L}+\frac{L^2}{2^d}$. For the two user set, let $k_1$ and $k_2$ be the number of reports attributed to the first and second user respectively (so that $k_1+k_2=k$). The probability of having two identical reports in the first user's reports is $\binom{k_1}{2}/L$, and similarly $\binom{k_2}{2}/L$ for the second user's reports. The probability of a first user's report being equal to a second user's report is upper bounded by the probability that both users choose an identical vector for one of their stored reports. This probability is upper bounded by $\frac{L^2}{2^d}$. This results in the probability of having two identical reports if the reports originated from two distinct users being at most $\frac{1}{L}\cdot\binom{k_1}{2}+\frac{1}{L}\cdot\binom{k_2}{2}+\frac{L^2}{2^d}\le\frac{k^2}{2L}+\frac{L^2}{2^d}$.
	\end{proof}
\end{document}